\documentclass[onecolumn, a4size, 11pt]{IEEEtran}
\usepackage{amsmath}
\usepackage{amssymb}
\usepackage{amsfonts}
\usepackage{graphicx}
\usepackage{epsfig}
\usepackage{subfigure}
\usepackage{psfrag}
\usepackage{xcolor}
\usepackage{cases}

\title{Multi-Antenna Wireless Powered Communication with Energy Beamforming \footnote {L. Liu is with the Department of Electrical
and Computer Engineering, National University of Singapore
(e-mail:liu\_liang@nus.edu.sg).}\footnote{R. Zhang is with the
Department of Electrical and Computer Engineering, National
University of Singapore (e-mail:elezhang@nus.edu.sg). He is also
with the Institute for Infocomm Research, A*STAR, Singapore.}
\footnote{K. C. Chua is with the Department of Electrical and
Computer Engineering, National University of Singapore
(e-mail:eleckc@nus.edu.sg).}}

\author{Liang Liu, Rui Zhang, and Kee-Chaing Chua}

\begin{document}
\maketitle \thispagestyle{empty} \vspace{-0.3in}

\begin{abstract}
The newly emerging wireless powered communication networks (WPCNs) have recently drawn significant attention, where radio signals are used to power wireless terminals for information transmission. In this paper, we study a WPCN where one multi-antenna access point (AP) coordinates energy transfer and information transfer to/from a set of single-antenna users. A harvest-then-transmit protocol is assumed where the AP first broadcasts wireless power to all users via energy beamforming in the downlink (DL), and then the users send their independent information to the AP simultaneously in the uplink (UL) using their harvested energy. To optimize the users' throughput and yet guarantee their rate fairness, we maximize the minimum throughput among all users by a joint design of the DL-UL time allocation, the DL energy beamforming, and the UL transmit power allocation plus receive beamforming. We solve this non-convex problem optimally by two steps. First, we fix the DL-UL time allocation and obtain the optimal DL energy beamforming, UL power allocation and receive beamforming to maximize the minimum signal-to-interference-plus-noise ratio (SINR) of all users. This problem is shown to be in general non-convex; however, we convert it equivalently to a spectral radius minimization problem, which can be solved efficiently by applying the alternating optimization based on the non-negative matrix theory. Then, the optimal time allocation is found by a one-dimension search to maximize the minimum rate of all users. Furthermore, two suboptimal designs of lower complexity are proposed, and their throughput performance is compared against that of the optimal solution.
\end{abstract}


\setlength{\baselineskip}{1.3\baselineskip}
\newtheorem{definition}{\underline{Definition}}[section]
\newtheorem{fact}{Fact}
\newtheorem{assumption}{Assumption}
\newtheorem{theorem}{\underline{Theorem}}[section]
\newtheorem{lemma}{\underline{Lemma}}[section]
\newtheorem{corollary}{\underline{Corollary}}[section]
\newtheorem{proposition}{\underline{Proposition}}[section]
\newtheorem{example}{\underline{Example}}[section]
\newtheorem{remark}{\underline{Remark}}[section]
\newtheorem{algorithm}{\underline{Algorithm}}[section]
\newcommand{\mv}[1]{\mbox{\boldmath{$ #1 $}}}

\section{Introduction}\label{eqn:Introduction}
Recently, energy harvesting has become an appealing solution to prolong the lifetime of energy constrained wireless networks such as device centric or sensor based wireless networks. In particular, radio frequency (RF) signals radiated by ambient transmitters is a viable new source for wireless energy harvesting. As a result, the wireless powered communication network (WPCN) has drawn an upsurge of interests, where RF signals are used to wirelessly power user terminals for communication. A typical WPCN model is shown in Fig. \ref{fig1} \cite{Rui11}, where an access point (AP) with constant power supply coordinates the downlink (DL) wireless information and energy transfer to a set of distributed user terminals that do not have embedded energy sources, as well as the wireless powered information transmission from the users in the uplink (UL).

\begin{figure}
\centering
 \epsfxsize=0.8\linewidth
    \includegraphics[width=10cm]{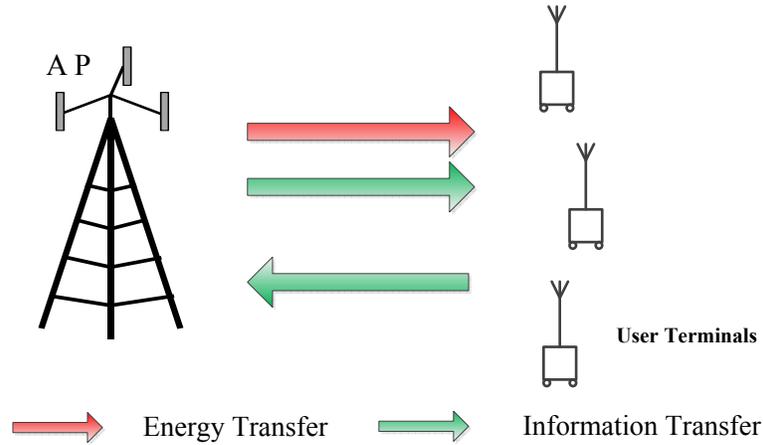}
\caption{A general wireless powered communication network (WPCN) with downlink (DL) information and energy transfer and uplink (UL) information transfer.}
\label{fig1}
\end{figure}

It is worth noting that the DL simultaneous wireless information and power transfer (SWIPT) in WPCNs has been recently studied in the literature (see e.g. \cite{Rui11}-\cite{RuiTWC}), where the achievable information versus energy transmission trade-offs were characterized under different channel setups. However, the above works have not addressed the joint design of DL energy transfer and UL information transmission in WPCNs, which is another interesting problem to investigate even by ignoring the DL information transmission for the purpose of exposition. In \cite{RuiGlobecom}, a WPCN with single-antenna AP and users has been studied for joint DL energy transfer and UL information transmission. A ``harvest-then-transmit'' protocol was proposed in \cite{RuiGlobecom} where the users first harvest energy from the signals broadcast by the AP in the DL, and then use their harvested energy to send independent information to the AP in the UL based on time-division-multiple-access (TDMA). The orthogonal time allocations for the DL energy transfer and UL information transmissions of all users are jointly optimized to maximize the network throughput. Furthermore, an interesting ``doubly near-far'' phenomenon was revealed in \cite{RuiGlobecom}, where a far user from the AP, which receives less power than a near user in the DL energy transfer, also suffers from more signal power attenuation in the UL information transmission due to pass loss.

\begin{figure}
\centering
 \epsfxsize=0.8\linewidth
    \includegraphics[width=12cm]{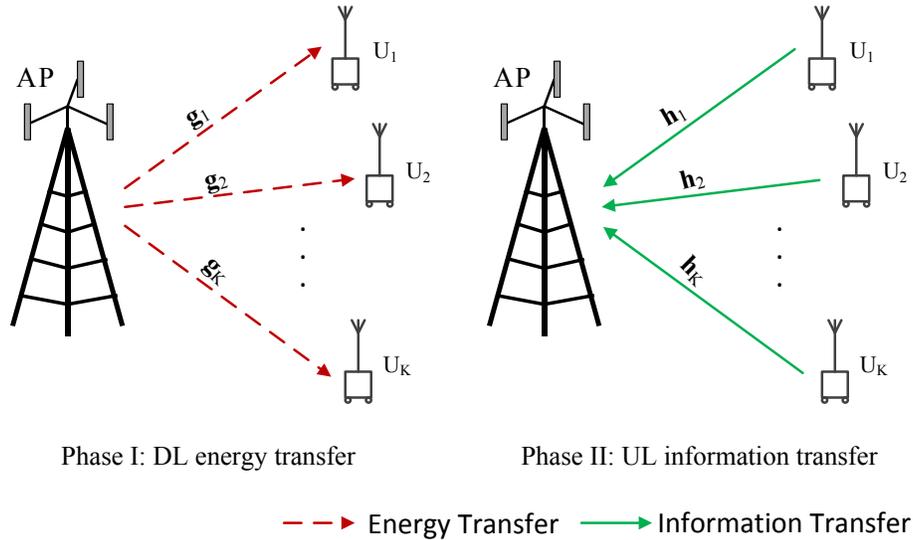}
\caption{A multi-antenna WPCN with DL energy transfer and UL
information transfer.}
\label{fig2}
\end{figure}

In this paper, we extend the study of \cite{RuiGlobecom} to WPCNs with the multi-antenna AP, as shown in Fig. \ref{fig2}. When the AP is equipped with multiple antennas, the amount of energy transferred to different users in the DL can be controlled by designing different energy beamforming weights at the AP, while in the UL all users can transmit information to the AP simultaneously via space-division-multiple-access (SDMA), which thus has higher spectrum efficiency than orthogonal user transmissions in TDMA as considered in \cite{RuiGlobecom}. To overcome the doubly near-far problem, similar to \cite{RuiGlobecom}, we maximize the minimum UL throughput among all users by a joint optimization of the DL-UL time allocation, the DL energy beamforming, and the UL transmit power allocation plus receive beamforming. First, we assume that the optimal linear minimum-mean-square-error (MMSE) based receiver is employed at the AP for UL information transmission, which results in a non-convex problem. We solve this problem optimally by two steps: First, we fix the DL-UL time allocation and obtain the corresponding optimal DL energy beamforming, UL power allocation and receive beamforming solution; then, the problem is solved by a one-dimension search over the optimal time allocation. Particularly, for the joint DL energy beamforming and UL power allocation plus receive beamforming optimization, it is shown that this problem is in general non-convex. However, we establish its equivalence to a spectral radius minimization problem, which is then solved globally optimally by applying the alternating optimization technique \cite{Schubert04} based on the non-negative matrix theory \cite{Horn85}, \cite{Seneta81}. Notice that the non-negative matrix theory has been applied in the literature to the UL multiuser information transmission with transmit power control and receive beamforming (see e.g. \cite{Schubert04}, \cite{Zhanglan08}, \textcolor{red}{\cite{Tan13}} and the references therein). Therefore, our proposed algorithm in this case can be viewed as an extension of the above works to the case with jointly optimizing the DL energy beamforming for wireless power transfer. It is also worth pointing out that in conventional multi-antenna wireless networks with both the UL and DL information transmissions, a useful tool that has been successfully applied to solve many non-convex design problems is the so-called UL-DL duality \cite{Schubert04}, \textcolor{red}{\cite{Tan13}}-\cite{Zhanglan13}. Different from this conventional setup, in this paper we explore another interesting new relationship between the DL and UL transmissions in a WPCN with coupled DL energy transfer and UL information transmission optimization. Finally, to reduce the complexity of the optimal solution, we propose two suboptimal designs employing the zero-forcing (ZF) based receive beamforming in the UL information transmission.

The rest of this paper is organized as follows. Section \ref{sec:System Model} presents the multi-antenna WPCN model with the harvest-then-transmit protocol. Section \ref{sec:Problem Formulation} formulates the minimum throughput maximization problem. Section \ref{sec:Optimal Solution} presents the optimal solution for this problem based on non-negative matrix theory. Section \ref{A ZF-Based Low-Complexity Suboptimal Solution} presents two suboptimal designs with lower complexity. Section \ref{sec:Numerical Results} provides numerical results to compare the performances of proposed solutions. Finally, Section \ref{sec:Conclusion} concludes the paper.

{\it Notation}: Scalars are denoted by lower-case letters, vectors
by bold-face lower-case letters, and matrices by
bold-face upper-case letters. $\mv{I}$ and $\mv{0}$  denote an
identity matrix and an all-zero matrix, respectively, with
appropriate dimensions. For a square matrix $\mv{S}$, ${\rm Tr}(\mv{S})$
denotes the trace of $\mv{S}$; $\mv{S}\succeq\mv{0}$ ($\mv{S}\preceq \mv{0}$) means that $\mv{S}$ is positive (negative) semi-definite. For a matrix
$\mv{M}$ of arbitrary size, $\mv{M}^{H}$ and ${\rm rank}(\mv{M})$ denote the
conjugate transpose and rank of $\mv{M}$, respectively. $E[\cdot]$ denotes the statistical expectation. The
distribution of a circularly symmetric complex Gaussian (CSCG) random vector with mean $\mv{x}$ and
covariance matrix $\mv{\Sigma}$ is denoted by
$\mathcal{CN}(\mv{x},\mv{\Sigma})$; and $\sim$ stands for
``distributed as''. $\mathbb{C}^{x \times y}$ denotes the space of
$x\times y$ complex matrices. $\|\mv{x}\|$ denotes the Euclidean norm of a complex vector
$\mv{x}$. For two real vectors
$\mv{x}$ and $\mv{y}$, $\mv{x}\geq \mv{y}$ means that $\mv{x}$ is
greater than or equal to $\mv{y}$ in a component-wise manner.

\section{System Model}\label{sec:System Model}

Consider a WPCN consisting of one AP and $K$ users, denoted by $U_k$, $1\leq k \leq K$, as shown in Fig. \ref{fig2}. It is assumed that the AP is equipped with $M>1$ antennas, while each $U_k$ is equipped with one antenna. The conjugated complex DL channel vector from the AP to $U_k$ and the reversed UL channel vector are denoted by $\mv{g}_k\in \mathbb{C}^{M\times 1}$ and $\mv{h}_k\in \mathbb{C}^{M\times 1}$, respectively. We assume that all channels follow independent quasi-static flat fading, where $\mv{g}_k$'s and $\mv{h}_k$'s remain
constant during one block transmission time, denoted by $T$, but in general can vary from block to block.\footnote{In practice, for the UL information transmission, the channels $\mv{h}_k$'s can be estimated by the AP based on the pilot signals sent by individual $U_k$'s, while for the DL power transfer, the channels $\mv{g}_k$'s can be obtained by the AP via, e.g., sending the pilot signal to all $U_k$'s and collecting channel estimation feedback from individual $U_k$'s. To focus on the performance upper bound, in this paper we assume that such channel knowledge is perfectly known at the AP for both DL and UL transmissions in each block.}

\begin{figure}
\centering
 \epsfxsize=0.8\linewidth
    \includegraphics[width=8cm]{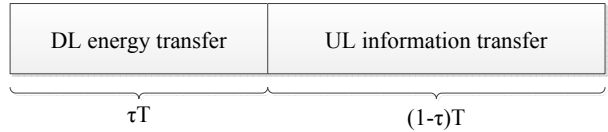}
\caption{The harvest-then-transmit protocol \cite{RuiGlobecom}.}
\label{fig3}
\end{figure}

In this paper, we assume that all $U_k$'s have no conventional energy supplies (e.g. fixed batteries) available and thus need to replenish energy from the signals sent by the AP in the DL. However, we assume that an energy storage device (ESD) in the form of rechargeable battery or super-capacitor is still equipped at each user terminal to store the energy harvested from received RF signals for future use. In particular, we adopt the ``harvest-then-transmit'' protocol proposed in \cite{RuiGlobecom}, as shown in Fig. \ref{fig3}, which is described as follows. In each block, during the first $\tau T$ ($0< \tau < 1$) amount of time, the AP broadcasts energy signals in the DL to transfer energy to all $U_k$'s simultaneously, while in the remaining $(1-\tau)T$ amount of time of the block, all $U_k$'s transmit their independent information to the AP simultaneously in the UL by SDMA using their harvested energy from the DL. For convenience, we normalize $T=1$ in the rest of this paper without loss of generality.

More specifically, during the DL phase, the AP transmits with $l$ energy beams to broadcast energy to all $U_k$'s, as shown in Fig. \ref{fig2}(a), where $l$ can be an arbitrary integer that is no larger than $M$. The baseband transmit signal $\mv{x}_0$ is thus expressed as
\begin{align}\label{eqn:energy signal}
\mv{x}_0=\sum\limits_{i=1}^l\mv{v}_is_i^{{\rm dl}},
\end{align}where $\mv{v}_i\in \mathbb{C}^{M\times 1}$ denotes the $i$th energy beam, and $s_i^{{\rm dl}}$ is its energy-carrying signal. It is assumed that $s_i^{{\rm dl}}$'s are independent and identically distributed (i.i.d.) random variables (RVs) with zero mean and unit variance. Then the transmit power of the AP in the DL can be expressed as $E[\|\mv{x}_0\|^2]=\sum_{i=1}^l\|\mv{v}_i\|^2$. Suppose that the AP has a transmit sum-power constraint $P_{{\rm sum}}$; thus, we have $\sum_{i=1}^l\|\mv{v}_i\|^2 \leq P_{{\rm sum}}$. The received signal in the DL at $U_k$ is then expressed as (by ignoring the receiver noise that is in practice negligible for energy receivers)
\begin{align}\label{eqn:received signal}
y_k=\mv{g}_k^H\mv{x}_0=\mv{g}_k^H\sum\limits_{i=1}^l\mv{v}_is_i^{{\rm dl}}, ~~~ k=1,\cdots,K.
\end{align}Due to the broadcast nature of wireless channels, the energy carried by all $l$ energy beams, i.e., $\mv{v}_i$'s ($i=1,\cdots,l$), can be harvested at each $U_k$. As a result, the harvested energy of $U_k$ in the DL can be expressed as
\begin{align}\label{eqn:energy}
E_k=\epsilon \tau E[|y_k|^2]=\epsilon \tau \sum\limits_{i=1}^l|\mv{g}_k^H\mv{v}_i|^2, ~~~ k=1,\cdots,K,
\end{align}where $0< \epsilon \leq 1$ denotes the energy harvesting efficiency at the receiver. Define $\mv{V}=\{\mv{v}_1,\cdots,\mv{v}_l\}$. Then, the average transmit power available for $U_k$ in the subsequent UL phase of information transmission is given by
\begin{align}\label{eq:available power}
\bar{P}_k(\mv{V},\tau) =\frac{E_k-E_k^{{\rm c}}}{1-\tau} =\frac{\epsilon \tau \sum\limits_{i=1}^l|\mv{g}_k^H\mv{v}_i|^2-E_k^{{\rm c}}}{1-\tau}, ~~~ k=1,\cdots,K,
\end{align}where $E_k^{{\rm c}}\geq 0$ denotes the circuit energy consumption at $U_k$ which is assumed to be constant over blocks. For convenience, we assume $E_k^{{\rm c}}=0$, $\forall k$, in the sequel to focus on transmit power for UL information transmission. Notice that thanks to multiple antennas equipped at the AP, we can schedule the UL transmit power at each $U_k$ via a proper selection of the DL energy beams in $\mv{V}$, which is not possible in a single-input single-output (SISO) WPCN with single-antenna AP as considered in \cite{RuiGlobecom}.

Next, in the UL phase, each $U_k$ utilizes its harvested energy in the previous DL phase to transmit information to the AP, as shown in Fig. \ref{fig2}(b). The transmit signal of $U_k$ in the UL is then expressed as
\begin{align}
x_k=\sqrt{p_k}s_k^{{\rm ul}}, ~~~ k=1,\cdots,K,
\end{align}where $s_k^{{\rm ul}}$'s denote the information-carrying signals of $U_k$'s, which are assumed to be i.i.d. circularly symmetric complex Gaussian (CSCG) RVs with zero mean and unit variance, denoted by $s_k^{{\rm ul}}\sim \mathcal{CN}(0,1)$, $\forall k$, and $p_k$ denotes the transmit power of $U_k$. Note that $p_k\leq \bar{P}_k(\mv{V},\tau)$, $\forall k$. The received signal at the AP in the UL is then expressed as
\begin{align}
\mv{y}=\sum\limits_{k=1}^K\mv{h}_kx_k+\mv{n}=\sum\limits_{k=1}^K\mv{h}_k\sqrt{p_k}s_k^{{\rm ul}}+\mv{n},
\end{align}where $\mv{n}\in \mathbb{C}^{M\times 1}$ denotes the receiver additive white Gaussian noise (AWGN). It is assumed that $\mv{n}\sim \mathcal{CN}(\mv{0},\sigma^2\mv{I})$. In this paper, we assume that the AP employs linear receivers to decode $s_k^{{\rm ul}}$'s in the UL. Specifically, let $\mv{w}_k\in \mathbb{C}^{M\times 1}$ denote the receive beamforming vector for decoding $s_k^{{\rm ul}}$, $k=1,\cdots,K$. Define $\mv{p}=[p_1,\cdots,p_K]^T$ and $\mv{W}=\{\mv{w}_1,\cdots,\mv{w}_K\}$. Then, the signal-to-interference-plus-noise ratio (SINR) for decoding $U_k$'s signal is expressed as
\begin{align}\label{equ:SINR in
SIMO-IC}
\gamma_k(\mv{p},\mv{w}_k)=\frac{p_k\|\mv{w}_k^H\mv{
h}_k\|^2}{\mv{w}_k^H\left(\sum\limits_{j\neq k}p_j\mv{h}_j\mv{
h}_j^H+\sigma^2\mv{I}\right)\mv{w}_k}, ~~~ k=1,\cdots,K.
\end{align}Thus, the achievable rate (in bps/Hz) for the UL information transmission of $U_k$ can be expressed as
\begin{align}\label{eqn:throughput}
R_k =(1-\tau)\log_2(1+\gamma_k(\mv{p},\mv{w}_k)) =(1-\tau)\log_2\left(1+\frac{p_k\|\mv{w}_k^H\mv{
h}_k\|^2}{\mv{w}_k^H\left(\sum\limits_{j\neq k}p_j\mv{h}_j\mv{
h}_j^H+\sigma^2\mv{I}\right)\mv{w}_k}\right), ~~~ \forall k.
\end{align}Notice that there exists a non-trivial trade-off in determining the optimal DL-UL time allocation $\tau$ to maximize $R_k$ since to increase the transmit power $p_k$, more time should be allocated to DL energy transfer according to (\ref{eq:available power}), while this will reduce the UL information transmission time from (\ref{eqn:throughput}).

\section{Problem Formulation}\label{sec:Problem Formulation}

In this paper, we are interested in maximizing the minimum (max-min) throughput of all $U_k$'s in each block by jointly optimizing the time allocation $\tau$, the DL energy beams $\mv{V}$, the UL transmit power allocation $\mv{p}$ and receive beamforming vectors $\mv{W}$, i.e.,
\begin{align}\label{equ:minimum throughput maximization}\mathop{\mathtt{Maximize}}_{\tau,\mv{p},\mv{W},\mv{V}}
& ~~ \min\limits_{1\leq k \leq K}(1-\tau)\log_2\left(1+\gamma_k(\mv{p},\mv{w}_k)\right) \nonumber \\ \mathtt{Subject \ to} & ~~ 0< \tau < 1, \nonumber \\ & ~~ p_k\leq \bar{P}_k(\mv{V},\tau), ~ \forall k, \nonumber \\ & ~~ \sum\limits_{i=1}^l\|\mv{v}_i\|^2\leq P_{{\rm sum}}.\end{align}It is worth noting that the number of energy beams, i.e., $l$, is a design variable in problem (\ref{equ:minimum throughput maximization}). After the DL energy beamforming solution $\mv{V}$ is obtained, we can set the optimal value of $l$ as the number of columns in $\mv{V}$.

Problem (\ref{equ:minimum throughput maximization}) is non-convex due to the coupled design variables in the objective function as well as the UL transmit power constraints. Note that if we fix $\tau=\bar{\tau}$ and $\mv{V}=\bar{\mv{V}}$, then problem (\ref{equ:minimum throughput maximization}) reduces to the following UL SINR balancing problem with the users' individual power constraints.
\begin{align}\label{equ:sinr balancing}\mathop{\mathtt{Maximize}}_{\mv{p},\mv{W}}
& ~~ \min\limits_{1\leq k \leq K}\gamma_k(\mv{p},\mv{w}_k) \nonumber \\ \mathtt{Subject \ to} & ~~ p_k\leq \bar{P}_k(\bar{\mv{V}},\bar{\tau}), ~ \forall k. \end{align}The above problem has been solved in the literature. For example, in \cite{Zhanglan08} problem (\ref{equ:sinr balancing}) was decoupled into $K$ subproblems, each with one individual user power constraint and thus solvable by the non-negative matrix theory based algorithm proposed in \cite{Schubert04}. In the following two sections, we propose both optimal and suboptimal algorithms to solve problem (\ref{equ:minimum throughput maximization}), respectively.

\section{Optimal Solution}\label{sec:Optimal Solution}

In this section, we propose to solve problem (\ref{equ:minimum throughput maximization}) optimally via a two-step procedure as follows. First, by fixing $\tau=\bar{\tau}$, $0<\bar{\tau}<1$, problem (\ref{equ:minimum throughput maximization}) reduces to the following problem.\begin{align}\label{equ:SINR balancing problem for SIMO-IC
with individual power constraint}\mathop{\mathtt{Maximize}}_{\mv{p},\mv{W},\mv{V}}
 & ~~ \min\limits_{1\leq k \leq K}\gamma_k(\mv{p},\mv{w}_k) \nonumber \\ \mathtt{Subject \ to} & ~~ p_k\leq \bar{P}_k(\mv{V},\bar{\tau}), ~ \forall k, \nonumber \\ & ~~ \sum\limits_{i=1}^l\|\mv{v}_i\|^2\leq P_{{\rm sum}}.\end{align}Let $g(\bar{\tau})$ denote the optimal value of problem (\ref{equ:SINR balancing problem for SIMO-IC
with individual power constraint}) with any given $\bar{\tau}$. The optimal value of problem (\ref{equ:minimum throughput maximization}) can then be obtained as \begin{align}\label{eqn:optimal time allocation}
R^\ast=\max\limits_{0< \bar{\tau} < 1} (1-\bar{\tau})\log_2(1+g(\bar{\tau})).\end{align}To summarize, problem (\ref{equ:minimum throughput maximization}) can be solved in the following two steps: First, given any $\bar{\tau}$, we solve problem (\ref{equ:SINR balancing problem for SIMO-IC
with individual power constraint}) to find $g(\bar{\tau})$; then, we solve problem (\ref{eqn:optimal time allocation}) to find the optimal $\bar{\tau}^\ast$ by a simple one-dimension search over $0< \bar{\tau} < 1$. In the rest of this section, we thus focus on solving problem (\ref{equ:SINR balancing problem for SIMO-IC
with individual power constraint}) with given $\bar{\tau}$. It is worth noting that as will be shown later in the numerical results in Section \ref{sec:Numerical Results}, with the optimal solution to problem (\ref{equ:SINR balancing problem for SIMO-IC
with individual power constraint}) for certain $\bar{\tau}$, denoted by $(\mv{p}^\ast,\mv{W}^\ast,\mv{V}^\ast)$, the users' individual power constraints in (\ref{equ:SINR balancing problem for SIMO-IC
with individual power constraint}) are not necessarily all tight, i.e., there may exist some $k$'s such that $p_k^\ast<\bar{P}_k(\mv{V}^\ast,\bar{\tau})$. This indicates that power control is in general needed in the UL information transmission since the optimal strategy for each user is not to always transmit with its maximum available power using the harvested energy from the DL power transfer.

By introducing a common SINR requirement $\gamma$ for all $U_k$'s, problem (\ref{equ:SINR balancing problem for SIMO-IC
with individual power constraint}) can be reformulated as the following problem.
\begin{align}\label{equ:common throughput}\mathop{\mathtt{Maximize}}_{\mv{p},\mv{W},\mv{V},\gamma}
& ~~ \gamma \nonumber \\ \mathtt{Subject \ to} & ~~ \gamma_k(\mv{p},\mv{w}_k)\geq \gamma, ~ \forall k, \nonumber \\ & ~~ p_k\leq \bar{P}_k(\mv{V},\bar{\tau}), ~ \forall k, \nonumber \\  & ~~ \sum\limits_{i=1}^l\|\mv{v}_i\|^2\leq P_{{\rm sum}}.\end{align}Note that even if we fix $\mv{V}=\bar{\mv{V}}$ in problem (\ref{equ:common throughput}), which reduces to the well-known SINR balancing problem given in (\ref{equ:sinr balancing}), this problem in general is still non-convex over $\mv{p}$, $\mv{W}$ and $\gamma$, and as a result its optimal solution cannot be obtained by convex optimization techniques \cite{Boyd04}. However, the non-negative matrix theory \cite{Horn85}, \cite{Seneta81} has been used in e.g., \cite{Schubert04}, \cite{Zhanglan08}, and \cite{Tan13} to obtain the optimal solution to problem (\ref{equ:sinr balancing}). By extending the results in \cite{Schubert04}, \cite{Zhanglan08}, and \cite{Tan13}, in the following we present an efficient algorithm to solve problem (\ref{equ:common throughput}) with the joint DL energy beamforming optimization based on the non-negative matrix theory.

First, we transform the SINR balancing problem given in
(\ref{equ:common throughput}) into an equivalent spectral radius minimization problem, where the spectral radius of a matrix $\mv{B}$, denoted by $\rho(\mv{B})$, is defined as its maximum eigenvalue in absolute value \cite{Horn85}, \cite{Seneta81}. Define $\mv{D}(\mv{W})={\rm Diag}\left\{\frac{1}{\|\mv{w}_1^H\mv{
h}_1\|^2},\cdots,\frac{1}{\|\mv{w}_K^H\mv{
h}_K\|^2}\right\}$, $\mv \sigma(\mv{W})=[(1-\bar{\tau})\sigma^2\|\mv{
w}_1\|^2,\cdots,(1-\bar{\tau})\sigma^2\|\mv{w}_K\|^2]^T$, and the $K
\times K$ non-negative matrix $\mv \Psi(\mv{W})$ as
\begin{align*}[\Psi(\mv{W})]_{k,j}=\left\{\begin{array}{ll}\|\mv{
w}_k^H\mv{h}_j\|^2, & k\neq j\\ 0, &
k=j,\end{array}\right.\end{align*}where $[\Psi(\mv{W})]_{k,j}$ denotes the entry on the $k$th row and $j$th column of $\Psi(\mv{W})$. Furthermore, define
\begin{align}\label{eqn:matrix Ak}\mv{A}_k(\mv{
W},\mv{V})=\left(\begin{array}{cc}\mv{D}(\mv{W})\mv \Psi(\mv{W}) & \mv{
D}(\mv{W})\mv \sigma(\mv{W})\\ \frac{\mv{e}_k^T\mv{D}(\mv{W})\mv
\Psi(\mv{W})}{\bar{P}_k(\mv{V},\bar{\tau})} & \frac{\mv{e}_k^T\mv{D}(\mv{W})\mv
\sigma(\mv{W})}{\bar{P}_k(\mv{V},\bar{\tau})}\end{array}\right), ~ \forall k, \end{align}where $\mv{e}_k\in \mathbb{C}^{K\times 1}$ denotes a vector with its $k$th component being $1$, and all other components being $0$. Then we have the following theorem.

\begin{theorem}\label{theorem1}
Given any feasible $\mv{W}$ and $\mv{V}$, the optimal SINR balancing solution to problem (\ref{equ:common throughput}) can be expressed as
\begin{align}\label{eqn:optimal SINR balancing value}
\gamma(\mv{W},\mv{V})=\frac{1}{\max\limits_{1\leq k \leq K}\rho(\mv{A}_k(\mv{W},\mv{V}))}.
\end{align}Furthermore, define $k^\ast=\arg\max\limits_{1\leq k \leq K}\rho(\mv{A}_k(\mv{W},\mv{V}))$, and $\left(\begin{array}{c}\mv{p} \\ 1 \end{array} \right)$ as the dominant eigenvector of $\mv{A}_{k^\ast}(\mv{W},\mv{V})$ corresponding to $\rho(\mv{A}_{k^\ast}(\mv{W},\mv{V}))$, then $\mv{p}$ is the optimal power solution to problem (\ref{equ:common throughput}) to achieve $\gamma(\mv{W},\mv{V})$ given $\mv{W}$ and $\mv{V}$.
\end{theorem}

\begin{proof}
Please refer to Appendix \ref{appendix1}.
\end{proof}

Theorem \ref{theorem1} implies that problem (\ref{equ:common throughput}) is equivalent to the following spectral radius minimization problem.
\begin{align}\label{eqn:equivalent problem}\mathop{\mathtt{Minimize}}_{\mv{W},\mv{V}}
& ~~ \max\limits_{1\leq k \leq K}\rho(\mv{A}_k(\mv{W},\mv{V})) \nonumber \\ \mathtt{Subject \ to} & ~~ \sum\limits_{i=1}^l\|\mv{v}_i\|^2\leq P_{{\rm sum}}.\end{align}

Next, we propose an iterative algorithm to solve problem (\ref{eqn:equivalent problem}) by applying the alternating optimization technique \cite{Schubert04}. Specifically, by fixing the UL receive beamforming $\mv{W}=\bar{\mv{W}}$, we first optimize the DL energy beamforming $\mv{V}$ by solving the following DL problem:
\begin{align}\label{eqn:downlink}\mathop{\mathtt{Minimize}}_{\mv{V}}
& ~~ \max\limits_{1\leq k \leq K}\rho(\mv{A}_k(\bar{\mv{W}},\mv{V})) \nonumber \\ \mathtt{Subject \ to} & ~~ \sum\limits_{i=1}^l\|\mv{v}_i\|^2\leq P_{{\rm sum}}.\end{align}Let $\bar{\mv{V}}$ denote the optimal solution to problem (\ref{eqn:downlink}), then by fixing $\mv{V}=\bar{\mv{V}}$, we optimize $\mv{W}$ by solving the following UL problem:
\begin{align}\label{eqn:uplink}\mathop{\mathtt{Minimize}}_{\mv{W}}
~~ \max\limits_{1\leq k \leq K}\rho(\mv{A}_k(\mv{W},\bar{\mv{V}})). \end{align}The above procedure is iterated until both $\mv{W}$ and $\mv{V}$ converge.

First, consider problem (\ref{eqn:downlink}). For convenience, define $\mv{X}(\bar{\mv{W}})=\mv{D}(\bar{\mv{W}})\Psi(\bar{\mv{W}})$, and $\mv{y}(\bar{\mv{W}})= \mv{D}(\bar{\mv{W}})\mv{\sigma}(\bar{\mv{W}})=[y_1(\bar{\mv{W}}),\cdots,y_K(\bar{\mv{W}})]^T$. Furthermore, let $[\mv{X}(\bar{\mv{W}})]_{i,j}$ denote the entry on the $i$th row and $j$th column of $\mv{X}(\bar{\mv{W}})$, and $[\mv{e}_k^T\mv{X}(\bar{\mv{W}})]_j$ denote the $j$th entry of $\mv{e}_k^T\mv{X}(\bar{\mv{W}})$, $\forall k$. Then we have the following proposition.

\begin{proposition}\label{proposition1}
Problem (\ref{eqn:downlink}) is equivalent to the following problem:
\begin{align}\label{eqn:downlink equivalent convex problem}\mathop{\mathtt{Minimize}}_{\mv{S},\tilde{\mv{q}},\tilde{\theta}} & ~~ e^{(\tilde{\theta})} \nonumber \\
\mathtt{Subject \ to} & ~~ \sum\limits_{j=1}^K[\mv{X}(\bar{\mv{W}})]_{i,j}e^{\tilde{q}_j-\tilde{q}_i -\tilde{\theta}}+y_i(\bar{\mv{W}})e^{\tilde{q}_{K+1}-\tilde{q}_i -\tilde{\theta}} \leq 1, ~ 1\leq i \leq K, \nonumber \\ & ~~ \sum\limits_{j=1}^K[\mv{e}_k^T\mv{X}(\bar{\mv{W}})]_je^{\tilde{q}_j-\tilde{q}_{K+1} - \tilde{\theta}}+\mv{e}_k^T\mv{y}(\bar{\mv{W}})e^{-\tilde{\theta}} \leq \frac{\epsilon \bar{\tau}{\rm Tr}(\mv{G}_k\mv{S})}{1-\bar{\tau}}, ~ 1\leq k \leq K, \nonumber \\ & ~~ {\rm Tr}(\mv{S}) \leq P_{{\rm sum}}, \nonumber \\ & ~~ \mv{S}\succeq \mv{0}, \end{align}where $\mv{S}=\sum_{i=1}^l\mv{v}_i\mv{v}_i^H$, and $\mv{G}_k=\mv{g}_k\mv{g}_k^H$, $\forall k$.
\end{proposition}

\begin{proof}
Please refer to Appendix \ref{appendix2}.
\end{proof}

Thanks to the fact that $\mv{A}_k(\bar{\mv{W}},\mv{V})$'s are all non-negative matrices, problem (\ref{eqn:downlink equivalent convex problem}) is a convex optimization problem, which thus can be efficiently solved by CVX \cite{Boyd11}. Let $\bar{\mv{S}}$ denote the optimal covariance solution to problem (\ref{eqn:downlink equivalent convex problem}); then the optimal $l={\rm rank}(\bar{\mv{S}})$ number of DL energy beams, i.e., $\bar{\mv{V}}=\{\bar{\mv{v}}_1,\cdots,\bar{\mv{v}}_l\}$, for problem (\ref{eqn:downlink}) can be obtained by computing the eigenvalue decomposition (EVD) of $\bar{\mv{S}}$.

Next, consider problem (\ref{eqn:uplink}). Since this problem has been solved by \cite{Zhanglan08}, we refer the readers to the algorithm given in Table IV of \cite{Zhanglan08} for the solution.

Last, by iteratively solving problems (\ref{eqn:downlink}) and (\ref{eqn:uplink}), we can solve problem (\ref{eqn:equivalent problem}), for which the overall algorithm is summarized in Table \ref{table5}. Since the objective value of problem (\ref{eqn:equivalent problem}) is increased after each iteration, a monotonic convergence can be guaranteed for Algorithm \ref{table5}. However, since problem (\ref{eqn:equivalent problem}) is a non-convex optimization problem, in general whether the converged solution is globally optimal to problem (\ref{eqn:equivalent problem}) remains unknown. In the following theorem, we show the global convergence of Algorithm \ref{table5}.

 \begin{table}[htp]
\begin{center}
\caption{\textbf{Algorithm \ref{table5}}: Algorithm for Solving
Problem (\ref{eqn:equivalent problem})} \vspace{0.2cm}
 \hrule
\vspace{0.2cm}
\begin{itemize}
\item[a)] Initialize a feasible $\mv{V}^{(1)}$ and update $\mv{W}^{(1)}$ via the algorithm in Table IV of \cite{Zhanglan08}. Set $\rho^{(1)}=\max\limits_{1\leq k \leq K} \rho(\mv{A}_k(\mv{W}^{(1)},\mv{V}^{(1)}))$ and $n=1$;
\item[b)] {\bf repeat}
\begin{itemize}
\item[1)] $n=n+1$;
\item[2)] DL: fix $\mv{W}=\mv{W}^{(n-1)}$ and update $\mv{V}^{(n)}$ by solving problem (\ref{eqn:downlink equivalent convex problem});
\item[3)] UL: fix $\mv{V}=\mv{V}^{(n)}$ and update $\mv{W}^{(n)}$ via the algorithm in Table IV of \cite{Zhanglan08};
\item[4)] Set $\rho^{(n)}=\max\limits_{1\leq k \leq K} \rho(\mv{A}_k(\mv{W}^{(n)},\mv{V}^{(n)}))$;
\end{itemize}
\item[c)] {\bf until} $\rho^{(n-1)}-\rho^{(n)}<\varepsilon$, where $\varepsilon$ is a small positive number to control the algorithm accuracy.
\end{itemize}
\vspace{0.2cm} \hrule \label{table5}
\end{center}
\end{table}

\begin{theorem}\label{theorem2}
Algorithm \ref{table5} converges to the globally optimal solution to problem (\ref{eqn:equivalent problem}).
\end{theorem}

\begin{proof}
Please refer to Appendix \ref{appendix3}.
\end{proof}

Due to the equivalence between problems (\ref{equ:common throughput}) and (\ref{eqn:equivalent problem}) shown in Theorem \ref{theorem1}, Theorem \ref{theorem2} implies that we can apply Algorithm \ref{table5} to obtain the optimal solution to problem (\ref{equ:common throughput}). Let $\mv{W}^\ast$ and $\mv{V}^\ast$ denote the optimal solution to problem (\ref{eqn:equivalent problem}) obtained by Algorithm \ref{table5}. We define $k^\ast=\arg \max\limits_{1\leq k \leq K} \rho(\mv{A}_k(\mv{W}^\ast,\mv{V}^\ast))$. Then, according to Theorem \ref{theorem1}, the optimal value of problem (\ref{equ:common throughput}), $\gamma^\ast$, is equal to $\frac{1}{\rho(\mv{A}_{k^\ast}(\mv{W}^\ast,\mv{V}^\ast))}$, and the optimal power solution $\mv{p}^\ast$ can be obtained from the dominant eigenvector of $\mv{A}_{k^\ast}(\mv{W}^\ast,\mv{V}^\ast)$, i.e., $\left(\begin{array}{c}\mv{p}^\ast \\ 1 \end{array} \right)$.

\section{Suboptimal Design}\label{A ZF-Based Low-Complexity Suboptimal Solution}
In the previous section, we propose the optimal algorithm to solve problem (\ref{equ:minimum throughput maximization}) based on the techniques of alternating optimization and non-negative matrix theory. Note that the optimal algorithm requires a joint optimization of the DL energy beams $\mv{V}$ and the UL transmit power allocation $\mv{p}$ plus receive beamforming $\mv{W}$. Moreover, the optimal time allocation for $\tau$ needs to be obtained by an exhaustive search. In this section, we propose two suboptimal solutions for problem (\ref{equ:minimum throughput maximization}) under the assumption that the number of users is no larger than that of antennas at the AP, i.e., $K\leq M$; hence, in the UL, the AP can employ the suboptimal ZF-based receivers (instead of MMSE-based receivers in the optimal algorithms) to completely eliminate the inter-user interference, which simplifies the design as shown next.

Define $\mv{H}_{-k}=[\mv{h}_1,\cdots,\mv{h}_{k-1},\mv{h}_{k+1},\cdots,\mv{h}_K]^H$, $k=1,\cdots,K$, which constitutes all the UL channels except $\mv{h}_k$. Then with ZF-based receivers in the UL, we aim to solve problem (\ref{equ:minimum throughput maximization}) with the additional constraints: $\mv{H}_{-k}\mv{w}_k=\mv{0}$, $\forall k$. Let the singular value decomposition (SVD) of $\mv{H}_{-k}$ be denoted as
\begin{align}
\mv{H}_{-k}=\mv{X}_k\mv{\Lambda}_k\mv{Y}_k^H=\mv{X}_k\mv{\Lambda}_k[\bar{\mv{Y}}_k \ \tilde{\mv{Y}}_k]^H,
\end{align}where $\mv{X}_k\in \mathbb{C}^{(K-1)\times (K-1)}$ and $\mv{Y}_k\in \mathbb{C}^{M\times M}$ are unitary matrices, and $\mv{\Lambda}_k$ is a $(K-1)\times M$ rectangular diagonal matrix. Furthermore, $\bar{\mv{Y}}_k\in \mathbb{C}^{M\times (K-1)}$ and $\tilde{\mv{Y}}_k\in \mathbb{C}^{M\times (M-K+1)}$ consist of the first $K-1$ and the last $M-K+1$ right singular vectors of $\mv{H}_{-k}$, respectively. Note that $\tilde{\mv{Y}}_k$ forms an orthogonal basis for the null space of $\mv{H}_{-k}$, thus $\mv{w}_k$ must be in the following form: $\mv{w}_k=\tilde{\mv{Y}}_k\tilde{\mv{w}}_k$, $\forall k$, where $\tilde{\mv{w}}_k$ is an arbitrary $(M-K+1)\times 1$ complex vector of unit norm. It can be shown that to maximize the rate of $U_k$, $\tilde{\mv{w}}_k$ should be aligned to the same direction as the equivalent channel $\tilde{\mv{Y}}_k^H\mv{h}_k$. Thus, we have\begin{align}\label{eqn:opt ZF}
\mv{w}_k^{{\rm ZF}}=\frac{\tilde{\mv{Y}}_k\tilde{\mv{Y}}_k^H\mv{h}_k}{\|\tilde{\mv{Y}}_k^H\mv{h}_k\|}, ~~ k=1,\cdots,K.
\end{align}Note that unlike the MMSE-based receivers in Section \ref{sec:Optimal Solution}, the above ZF receivers are not related to $\mv{p}$ and hence do not depend on $\mv{V}$ and $\tau$.

With the ZF receivers given in (\ref{eqn:opt ZF}), the throughput of $U_k$ given in (\ref{eqn:throughput}) reduces to
\begin{align}\label{eqn:throughput ZF}
R_k^{{\rm ZF}}=(1-\tau)\log_2\left(1+\frac{\tilde{h}_kp_k}{\sigma^2}\right), ~~ k=1,\cdots,K,
\end{align}where $\tilde{h}_k=\|\tilde{\mv{Y}}_k^H\mv{h}_k\|^2$ denotes the power of the equivalent UL channel for $U_k$. Based on the achievable rate expression given in (\ref{eqn:throughput ZF}) with ZF receive beamforming, we further propose two suboptimal solutions to obtain $\tau$, $\mv{p}$, and $\mv{V}$ for problem (\ref{equ:minimum throughput maximization}) in the following two subsections, respectively.

\subsection{Suboptimal Solution 1}

With (\ref{eqn:throughput ZF}), problem (\ref{equ:minimum throughput maximization}) reduces to
\begin{align}\label{eqn:ZF problem}\mathop{\mathtt{Maximize}}_{\tau,\mv{p},\mv{V}}
& ~~ \min\limits_{1\leq k \leq K}(1-\tau)\log_2\left(1+\frac{\tilde{h}_kp_k}{\sigma^2}\right) \nonumber \\ \mathtt{Subject \ to} & ~~ 0< \tau < 1, \nonumber \\ & ~~ p_k\leq \bar{P}_k(\mv{V},\tau), ~~ \forall k, \nonumber \\ & ~~ \sum\limits_{i=1}^l\|\mv{v}_i\|^2\leq P_{{\rm sum}}.\end{align}Define $\tilde{p}_k=(1-\tau)p_k$, $\forall k$, and $\tilde{\mv{S}}=\tau \sum_{i=1}^l\mv{v}_i\mv{v}_i^H$. By introducing a common throughput requirement $\bar{R}$, problem (\ref{eqn:ZF problem}) can be transformed into the following equivalent problem.
\begin{align}\label{eqn:equivalent ZF problem}\mathop{\mathtt{Maximize}}_{\tau,\tilde{\mv{p}},\tilde{\mv{S}},\bar{R}}
& ~~ \bar{R} \nonumber \\ \mathtt{Subject \ to} & ~~ (1-\tau)\log_2\left(1+\frac{\tilde{h}_k\tilde{p}_k}{(1-\tau)\sigma^2}\right) \geq \bar{R}, ~~ \forall k, \nonumber \\ & ~~ 0< \tau < 1, \nonumber \\ & ~~ \tilde{p}_k\leq \epsilon {\rm Tr}(\mv{G}_k\tilde{\mv{S}}), ~~ \forall k, \nonumber \\ & ~~ {\rm Tr}(\tilde{\mv{S}})\leq \tau P_{{\rm sum}},\end{align}where $\tilde{\mv{p}}=\{\tilde{p}_1,\cdots,\tilde{p}_K\}$.

Problem (\ref{eqn:equivalent ZF problem}) can be shown to be convex, and thus it can be solved efficiently by e.g., the interior-point method \cite{Boyd04}. Let $\tau^{(1)}$, $\tilde{\mv{p}}^{(1)}$, $\tilde{\mv{S}}^{(1)}$ and $\bar{R}^{(1)}$ denote the optimal solution to problem (\ref{eqn:equivalent ZF problem}). Then the optimal power allocation solution to problem (\ref{eqn:ZF problem}) can be obtained as $p_k^{(1)}=\tilde{p}_k^{(1)}/(1-\tau^{(1)})$, and the optimal $l^{(1)}={\rm rank}(\tilde{\mv{S}}^{(1)})$ number of energy beams $\mv{v}_i^{(1)}$'s can be obtained by the EVD of $\tilde{\mv{S}}^{(1)}/\tau^{(1)}$.

\subsection{Suboptimal Solution 2}

Problem (\ref{eqn:ZF problem}) still requires a joint optimization of $\mv{V}$, $\tau$ and $\mv{p}$. To further reduce the complexity, in this subsection we propose another suboptimal solution for problem (\ref{eqn:ZF problem}) by separating the optimization of DL energy beamforming and UL power allocation. First, the DL energy beams $\mv{v}_i$'s are obtained by solving the following weighted sum-energy maximization problem.
\begin{align}\label{eqn:maximum energy}~\mathop{\mathtt{Maximize}}_{\mv{V}}
& ~~~  \sum\limits_{k=1}^K \alpha_k \epsilon \left(\sum\limits_{i=1}^l|\mv{g}_k^H\mv{v}_i|^2\right) \nonumber \\
\mathtt {Subject \ to} & ~~~
\sum\limits_{i=1}^l\|\mv{v}_i\|^2 \leq P_{{\rm sum}},
\end{align}where $\alpha_k\geq 0$ denotes the energy weight for $U_k$. Note that intuitively, to guarantee the rate fairness among the users, in the DL we should transfer more energy to users with weaker channels (e.g., more distant from the AP) by assigning them with higher energy weights. Therefore, we propose the following energy weight assignment rule that takes the doubly near-far effect into account: $\alpha_k=1/(\tilde{h}_k\|\mv{g}_k\|^2)$, $k=1,\cdots,K$. Let $\psi$ and $\mv{\eta}$ denote the maximum eigenvalue and its corresponding unit-norm eigenvector of the matrix $\sum_{k=1}^K\alpha_k\epsilon\mv{G}_k$, respectively. From \cite{RuiICCASP}, the optimal value of problem (\ref{eqn:maximum energy}) given a set of $\alpha_k$'s is then obtained as
$E_{{\rm max}}=\psi P_{{\rm sum}}$, which is achieved by $l^{(2)}=1$ and $\mv{v}_1^{(2)}=\sqrt{P_{{\rm sum}}}\mv{\eta}$, i.e., only one energy beam is used. Next, by substituting $\mv{v}_1^{(2)}$ into problem (\ref{eqn:ZF problem}), the corresponding optimal time allocation $\tau^{(2)}$ and power allocation $\mv{p}^{(2)}$ can be obtained by solving the following problem:
\begin{align}\label{eqn:equivalent subooptimal 2}\mathop{\mathtt{Maximize}}_{\tau,\tilde{\mv{p}},\bar{R}}
& ~~ \bar{R} \nonumber \\ \mathtt{Subject \ to} & ~~ (1-\tau)\log_2\left(1+\frac{\tilde{h}_k\tilde{p}_k}{(1-\tau)\sigma^2}\right) \geq \bar{R}, ~~ \forall k, \nonumber \\ & ~~ 0< \tau < 1, \nonumber \\ & ~~ \tilde{p}_k\leq \epsilon \tau\|\mv{g}_k^H\mv{v}_1^{(2)}\|^2, ~~ \forall k. \end{align}

It is worth noting that all $U_k$'s should transmit at full power in the UL in this case since they cause no interference to each other due to the ZF receivers used at the AP. As a result, without loss of generality we can substitute $\tilde{p}_k=\epsilon \tau \|\mv{g}_k^H\mv{v}_1^{(2)}\|^2$ into problem (\ref{eqn:equivalent subooptimal 2}) to remove the variable $\tilde{\mv{p}}$, which results in the following equivalent problem:
\begin{align}\label{eqn:equivalent subooptimal zf}\mathop{\mathtt{Maximize}}_{\tau,\bar{R}}
& ~~ \bar{R} \nonumber \\ \mathtt{Subject \ to} & ~~ (1-\tau)\log_2\left(1+\frac{\tilde{h}_k\epsilon \|\mv{g}_k^H\mv{v}_1^{(2)}\|^2\tau}{(1-\tau)\sigma^2}\right) \geq \bar{R}, ~~ \forall k, \nonumber \\ & ~~ 0< \tau < 1.\end{align}

It can be shown that $(1-\tau)\log_2\left(1+\tilde{h}_k\epsilon \|\mv{g}_k^H\mv{v}_1^{(2)}\|^2\tau/(1-\tau)\sigma^2\right)$ is a concave function over $0<\tau<1$, and thus problem (\ref{eqn:equivalent subooptimal zf}) is a convex optimization problem, which can be solved efficiently by the interior-point method \cite{Boyd04}. Alternatively, we can apply the bisection method \cite{Boyd04} to search for the optimal $\bar{R}$, while with given $\bar{R}$, the optimal time allocation $\tau$ can be efficiently obtained by solving a convex feasibility problem, for which the details are omitted here for brevity.

\section{Numerical Results}\label{sec:Numerical Results}

In this section, we provide numerical examples to validate our results. We consider a multi-antenna WPCN in which the AP is equipped with $M=6$ antennas, and there are $K=4$ users.\footnote{Note that $K\leq M$ holds in our example; thus, the two ZF receiver based solutions in Section \ref{A ZF-Based Low-Complexity Suboptimal Solution} are both feasible.} We set $P_{{\rm sum}}=1$Watt (W) or $30$dBm, $\epsilon=50\%$, and $\sigma^2=-50$dBm. The distance-dependent pass loss model is given by
\begin{align}\label{eqn:pass loss}
L_k=A_0\left(\frac{d_k}{d_0}\right)^{-\alpha}, ~~~ k=1,\cdots,K,
\end{align}where $A_0$ is set to be $10^{-3}$, $d_k$ denotes the distance between $U_k$ and AP, $d_0$ is a reference distance set to be $1$m, and $\alpha$ is the path loss exponent set to be $3$. Moreover, we assume that the channel reciprocity holds for the UL and DL channels, i.e., $\mv{h}_k=\mv{g}_k$, $\forall k$. The channel vectors $\mv{g}_k$'s are generated from independent Rician fading. Specifically, $\mv{g}_k$ is expressed as
\begin{equation}
\mv{g}_k = \sqrt{\frac{K_R}{1+K_R}}\mv{g}_k^{{\rm LOS}}+\sqrt{\frac{1}{1+K_R}}\mv{g}_k^{{\rm NLOS}}, ~~ \forall k,
\end{equation}where $\mv{g}_k^{{\rm LOS}}\in \mathbb{C}^{M\times 1}$ is the line of sight (LOS) deterministic component, $\mv{g}_k^{{\rm NLOS}}\in \mathbb{C}^{M\times 1}$ denotes the Rayleigh fading
component consisting of i.i.d. CSCG RVs with zero mean and unit covariance, and $K_R$ is the Rician factor set to be 3. Note that for the LOS component, we use the far-field uniform linear antenna array model with $\mv{g}_k^{{\rm LOS}}=[1~e^{j\theta_k}~e^{j2\theta_k}~\ldots~e^{j(N_t-1)\theta_k}~]^T$ with $\theta_k=-\frac{2\pi d^{{\rm an}} \sin(\varphi_k)}{\lambda}$, where $d^{{\rm an}}$ is the spacing between successive antenna elements at the AP, $\lambda$ is the carrier wavelength, and $\varphi_k$ is the direction of $U_k$ to the AP. We set $d^{{\rm an}}=\frac{\lambda}{2}$, and $\{\varphi_1,\varphi_2,\varphi_3,\varphi_4\}=\{-45^{\textrm{o}}, -15^{\textrm{o}}, 15^{\textrm{o}}, 45^{\textrm{o}}\}$. The average power of $\mv{g}_k$ is then normalized by $L_k$ in (\ref{eqn:pass loss}).

\subsection{Optimal Solution}

In this subsection, we investigate the performance of the optimal solution proposed in Section \ref{sec:Optimal Solution}. In this numerical result, we set $d_1=1$m, $d_2=1.4$m, $d_3=1.8$m, and $d_4=2$m. Specifically, the channels $\bar{\mv{H}}=[\mv{h}_1,\cdots,\mv{h}_4]$ and $\bar{\mv{G}}=[\mv{g}_1,\cdots,\mv{g}_4]$ are given by
\begin{align}
\bar{\mv{H}}&=\bar{\mv{G}}\nonumber \\ & =\left[\begin{array}{cccc}0.0082+0.0085i & 0.01371-0.0022i & 0.0133+0.0077i & 0.0081-0.0004i \\ 0.0021 + 0.0110i &  0.0383 + 0.0125i  & 0.0162 + 0.0061i &  0.0113 - 0.0051i \\ -0.0246 - 0.0104i &  0.0172 + 0.0271i  & 0.0236 + 0.0125i &  0.0003 - 0.0136i \\ -0.0184 - 0.0174i & -0.0364 + 0.0023i &  0.0194 + 0.0031i & -0.0131 - 0.0110i \\ 0.0411 + 0.0017i & -0.0371 - 0.0106i & -0.0032 - 0.0064i & -0.0161 + 0.0009i \\ -0.0002 + 0.0516i & -0.0172 - 0.0160i &  0.0202 - 0.0014i & -0.0151 + 0.0075i
        \end{array}\right].
\end{align}

First, we investigate the impact of $\tau$ on the max-min throughput among $U_k$'s. Let $R_{{\rm MMSE}}(\bar{\tau})$ denote the max-min throughput achieved by MMSE receivers given the time allocation $\bar{\tau}$. For the purpose of comparison, we also study the max-min throughput achieved by ZF receivers, denoted by $R_{{\rm ZF}}(\bar{\tau})$. Note that $R_{{\rm MMSE}}(\bar{\tau})=(1-\bar{\tau})\log_2(1+g(\bar{\tau}))$. Also note that $R_{{\rm ZF}}(\bar{\tau})$ can be obtained by solving problem (\ref{eqn:ZF problem}) with $\tau=\bar{\tau}$. Fig. \ref{fig4} shows $R_{{\rm MMSE}}(\bar{\tau})$ versus $R_{{\rm ZF}}(\bar{\tau})$ over $0<\bar{\tau}<1$. It is observed that both $R_{{\rm MMSE}}(\bar{\tau})$ and $R_{{\rm ZF}}(\bar{\tau})$ are first increasing and then decreasing over $\bar{\tau}$. The reason is as follows. It can be observed from (\ref{eqn:throughput}) that when $\bar{\tau}$ is small, the available transmit power for users given in (\ref{eq:available power}) is the dominant factor and thus increasing $\tau$ increases the DL energy transfer time and hence the UL transmit power and throughput. However, when $\bar{\tau}$ becomes large, the UL transmission time becomes the limiting factor and as a result increasing $\tau$ decreases the UL transmission time and thus the throughput. It is also observed that MMSE receiver achieves higher throughput than ZF receiver for any given $\bar{\tau}$.

\begin{figure}
\centering
 \epsfxsize=0.8\linewidth
    \includegraphics[width=10cm]{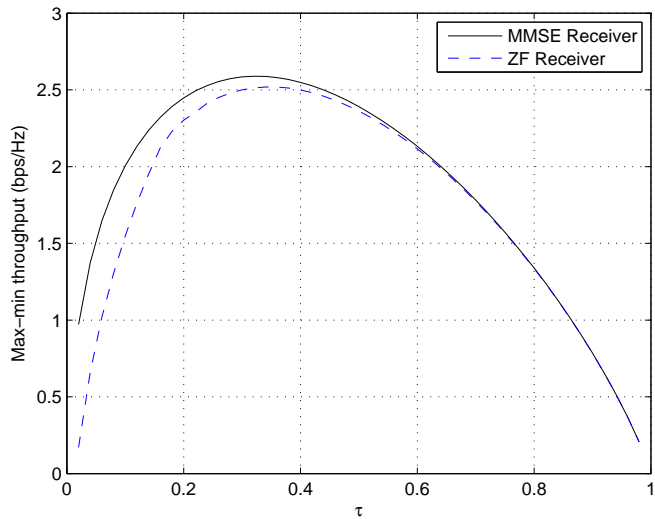}
\caption{$R_{{\rm MMSE}}(\bar{\tau})$ versus $R_{{\rm ZF}}(\bar{\tau})$.}
\label{fig4}
\end{figure}

\begin{figure}
\centering
 \epsfxsize=0.8\linewidth
    \includegraphics[width=10cm]{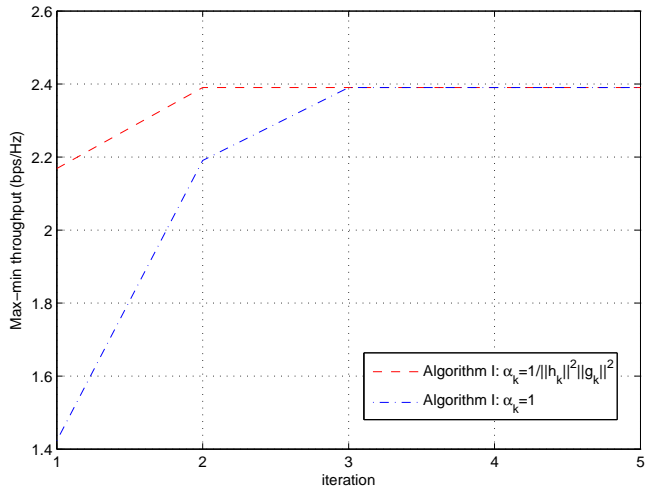}
\caption{Max-min throughput achieved by Algorithm \ref{table5} versus iteration when $\tau=0.5$ with different initial points.}
\label{fig6}
\end{figure}

Next, we study the performance of the optimal solutions to problem (\ref{equ:common throughput}) proposed in Section \ref{sec:Optimal Solution} with $\tau=0.5$. Fig. \ref{fig6} shows the convergence performance of Algorithm \ref{table5} with different initial points of $\mv{V}$. Specifically, two initial points of $\mv{V}$ are obtained by solving problem (\ref{eqn:maximum energy}) with $\alpha_k=1$ and $\alpha_k=1/\|\mv{h}_k\|^2\|\mv{g}_k\|^2$, $\forall k$, respectively. It is observed that Algorithm \ref{table5} does converge to the optimal solution in only $4$-$5$ iterations for both initial points. It is also observed that the initial point of $\mv{V}$ obtained by setting $\alpha_k=1/\|\mv{h}_k\|^2\|\mv{g}_k\|^2$, $\forall k$, in problem (\ref{eqn:maximum energy}) is better than that obtained by setting $\alpha_k=1$, $\forall k$, to make Algorithm \ref{table5} converge faster. The reason is as follows. When we fix $\alpha_k=1$, $\forall k$, in problem (\ref{eqn:maximum energy}), in the DL the users more far away from the AP tend to be allocated with less energy, i.e., incurring the doubly near-far effect in the WPCN. However, by setting $\alpha_k=1/\|\mv{h}_k\|^2\|\mv{g}_k\|^2$, $\forall k$, the users with poorer channels are assigned with higher priority in the DL power transfer, and thus have more transmit power in the UL information transmission.

\begin{table}
\caption{$\bar{P}_k(\mv{V}^\ast,\tau=0.5)$ versus $p_k^\ast$} \label{table1}
\begin{center}
\begin{tabular}{c|c|c}
\hline User Index $k$ & $\bar{P}_k(\mv{V}^\ast,\tau=0.5)$ (mW) & $p_k^\ast$ (mW) \\
\hline\hline
$1$ & $0.4913$ & $0.0846$ \\
$2$ & $0.6869$ & $0.0987$ \\
$3$ & $0.3168$ & $0.2547$ \\
$4$ &  $0.6199$ & $0.6199$ \\
 \hline
\end{tabular}
\end{center}
\end{table}

Furthermore, to illustrate whether power control is needed in the UL information transmission, i.e., each user transmits at maximum power or not using the energy harvested from the DL power transfer, we show the values of $\bar{P}_k(\mv{V}^\ast,\tau=0.5)$ versus $p_k^\ast$, $\forall k$, in Table \ref{table1}, where $\mv{p}^\ast$ and $\mv{V}^\ast$ denote the optimal solution to problem (\ref{equ:common throughput}) with $\tau=0.5$. It is observed that the three users that are nearer to the AP, i.e., $U_1$, $U_2$ and $U_3$, should not transmit at maximum power, and thus in general given the optimal DL energy beams $\mv{V}^\ast$, UL power control is needed to maximize the minimum SINR of all users in problem (\ref{equ:common throughput}).

\begin{figure}
\centering
 \epsfxsize=0.8\linewidth
    \includegraphics[width=10cm]{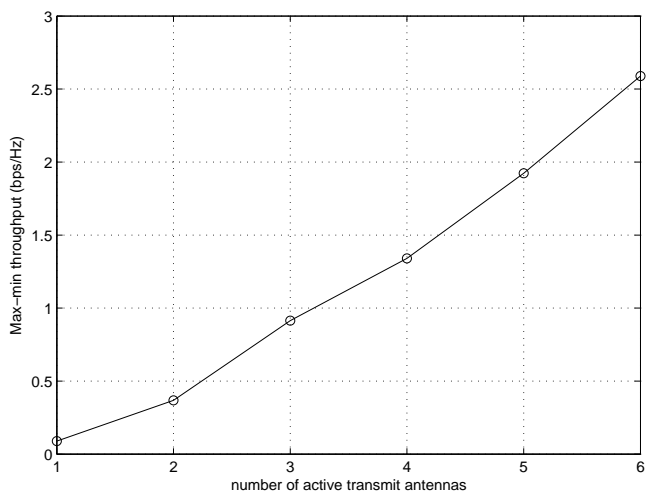}
\caption{Max-min throughput achieved by the optimal solution versus the number of active antennas at the AP.}
\label{fig7}
\end{figure}

Last, we study the impact of the number of antennas at the AP on the max-min throughput performance. In this example, we activate one more antenna among the $M=6$ antennas at each time. Fig. \ref{fig7} shows the max-min throughput achieved by the optimal solution in Section \ref{sec:Optimal Solution} versus the number of active antennas at the AP. Note that for the case when there is only one active antenna at the AP, since spatial transmit/receive beamforming cannot be utilized, we adopt the TDMA based solution proposed in \cite{RuiGlobecom} for the SISO WPCN. It is observed from Fig. \ref{fig7} that the max-min throughput increases significantly with the number of active antennas at the AP.

\subsection{Suboptimal Solution}

\begin{figure}
\centering
 \epsfxsize=0.8\linewidth
    \includegraphics[width=10cm]{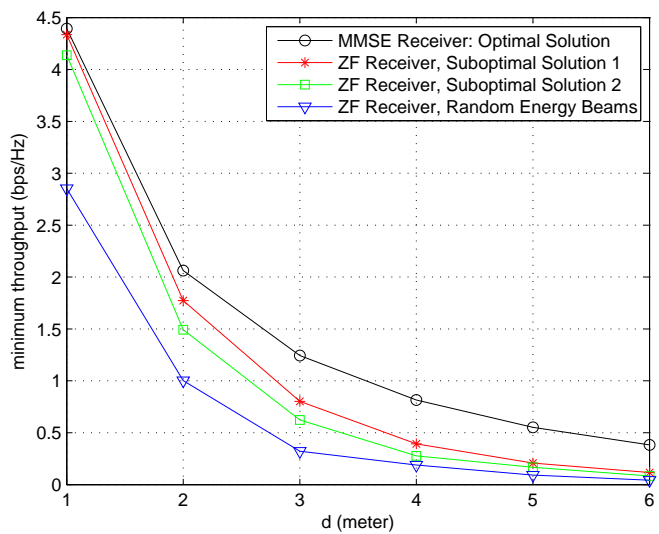}
\caption{Performance comparison between the optimal and suboptimal solutions.}
\label{fig5}
\end{figure}

In this subsection, we compare the max-min throughput by the optimal solution in Section \ref{sec:Optimal Solution} with MMSE receivers and the two suboptimal solutions in Section \ref{A ZF-Based Low-Complexity Suboptimal Solution} with ZF receivers. In this example, it is assumed that all users are of the same distance to the AP, i.e., $d_k=d$, $\forall k$. Fig. \ref{fig5} shows the max-min throughput over $d$. For the purpose of comparison, we also plot the max-min throughout achieved by solving problem (\ref{eqn:equivalent subooptimal zf}) where the energy beams $\mv{V}$ are randomly generated rather than obtained via solving problem (\ref{eqn:maximum energy}). It is observed that the throughput decays drastically as $d$ increases for all optimal and suboptimal solutions. It is also observed that for all values of $d$, the throughput by MMSE receiver outperforms those of the three suboptimal solutions by ZF receiver. However, when $d$ is small, it is observed that both Suboptimal Solutions 1 and 2 with ZF receiver achieve the throughput very close to the optimal solution with MMSE receiver. This is because in this case the available power for UL transmission is large for all $U_k$'s, and thus ZF receiver is asymptotically optimal with high signal-to-noise ratio (SNR). Furthermore, it is observed that with ZF receiver, Suboptimal Solution 2 performs very close to Suboptimal Solution 1, although it is based on separate optimizations of DL energy beamforming and UL power allocation to achieve lower complexity. However, if the energy beams are randomly generated instead of via solving problem (\ref{eqn:maximum energy}), there is a significant loss in the achieved max-min throughput observed with ZF receiver.

\section{Conclusion}\label{sec:Conclusion}
This paper has studied a wireless powered communication network (WPCN) with multi-antenna AP and single-antenna users. Under a harvest-then-transmit protocol, the minimum throughput among all users is maximized by a joint optimization of the DL-UL time allocation, DL energy beamforming, and UL transmit power allocation plus receive beamforming. We solve this problem optimally via a two-stage algorithm. First, we fix the DL-UL time allocation and propose an efficient algorithm to obtain the corresponding optimal DL energy beamforming and UL power allocation plus receive beamforming solution based on the techniques of alternating optimization and non-negative matrix theory. Then, the problem is solved by a one-dimension search over the optimal DL-UL time allocation. Furthermore, two suboptimal solutions of lower complexity are proposed with ZF based receive beamforming, and their performances are compared to the optimal solution.

\begin{appendix}

\subsection{Proof of Theorem \ref{theorem1}}\label{appendix1}

First, we have the following lemma.

\begin{lemma}\label{lemma5}
Given any receive beamforming
vectors $\mv{W}=\bar{\mv{W}}$ and energy beams $\mv{V}=\bar{\mv{V}}$, the
corresponding optimal power allocation $\bar{\mv{p}}$ and SINR balancing solution $\gamma(\bar{\mv{W}},\bar{\mv{V}})$ to problem (\ref{equ:common throughput}) must satisfy
the following two conditions:
\begin{itemize}
\item[1.] All $U_k$'s, $k=1,\cdots,K$, achieve the same SINR balancing value, i.e.,
\begin{align}\label{equ:equal sinr balancing level}
\gamma_k(\bar{\mv{p}},\bar{\mv{w}}_k)=\gamma(\bar{\mv{W}},\bar{\mv{V}}),
\ \forall k.
\end{align}
\item[2.] There exists at least an $U_{k^\ast}$ such that $\bar{p}_{k^\ast}=\bar{P}_{k^\ast}(\bar{\mv{V}},\bar{\tau})$.
\end{itemize}
\end{lemma}

\begin{proof}
First, we assume that with $\bar{\mv{p}}$, there exists an $U_{\bar{k}}$ such that $\gamma_{\bar{k}}(\bar{\mv{p}},\bar{\mv{w}}_{\bar{k}})>\gamma(\bar{\mv{W}},\bar{\mv{V}})$. Then, we can decrease the transmit power of $U_{\bar{k}}$ and at the same time keep the transmit power of all other $U_k$'s, $\forall k \neq \bar{k}$, unchanged such that $U_{\bar{k}}$'s SINR is reduced but still larger than $\gamma(\bar{\mv{W}},\bar{\mv{V}})$. Note that this will increase each of other $U_k$'s SINR, $\forall k \neq  \bar{k}$, to be larger than $\gamma(\bar{\mv{W}},\bar{\mv{V}})$, since the interference power from $U_{\bar{k}}$ is reduced. As a result, the minimum SINR of $U_k$'s must be larger than $\gamma(\bar{\mv{W}},\bar{\mv{V}})$ with the new constructed power allocation, which contradicts to the fact that $\bar{\mv{p}}$ is the optimal power solution to problem (\ref{equ:common throughput}). The first part of Lemma \ref{lemma5} is thus proved.

Next, we assume that with $\bar{\mv{p}}$, all the individual power constraints are not tight in (\ref{equ:common throughput}), i.e., $\bar{p}_{k}<\bar{P}_{k}(\bar{\mv{V}},\bar{\tau})$, $\forall k$. In this case, define $\alpha=\min_{1\leq k \leq K}\bar{P}_{k}(\bar{\mv{V}},\bar{\tau})/\bar{p}_{k}>1$. Then, consider the new power solution $\hat{\mv{p}}=\alpha \bar{\mv{p}}$, which satisfies all the individual power constraints in problem (\ref{equ:common throughput}). Since $\gamma_k(\beta\bar{\mv{p}},\bar{\mv{w}}_k)>\gamma_k(\bar{\mv{p}},\bar{\mv{w}}_k)$ holds $\forall \beta>1$, $\forall k$, the minimum SINR of all $U_k$'s must be increased with the new constructed power solution $\hat{\mv{p}}$, which contradicts to the fact that $\bar{\mv{p}}$ is the optimal power solution to problem (\ref{equ:common throughput}). The second part of Lemma \ref{lemma5} is thus proved.
\end{proof}

We can express (\ref{equ:equal sinr balancing
level}) for all $k$'s in the following matrix
form:\begin{align}\label{eqn1}\bar{\mv{p}}\frac{1}{\gamma(\bar{\mv{W}},\bar{\mv{V}})}=\mv{D}(\bar{\mv{W}}) \mv
\Psi(\bar{\mv{W}})\bar{\mv{p}}+\mv{D}(\bar{\mv{W}})\mv \sigma(\bar{\mv{W}}).\end{align}Therefore, given any $\mv{W}=\bar{\mv{W}}$ and $\mv{V}=\bar{\mv{V}}$, the optimal power allocation $\bar{\mv{p}}$ and SINR balancing solution $\gamma(\bar{\mv{W}},\bar{\mv{V}})$ to problem (\ref{equ:common throughput}) must satisfy
\begin{numcases}{}
(\ref{eqn1}), \nonumber \\
\bar{p}_k=\bar{P}_k(\bar{\mv{V}},\bar{\tau}), ~ k=k^\ast, \label{eqn2}\\
\bar{p}_k\leq \bar{P}_k(\bar{\mv{V}},\bar{\tau}), ~ \forall k\neq k^\ast. \label{eqn3}
\end{numcases}

The following lemma reveals one important property for the equations given in (\ref{eqn1}) and (\ref{eqn2}).
\begin{lemma}\label{lemma3}Given
any fixed $\bar{\mv{W}}$ and $\bar{\mv{V}}$, there exists a unique solution
$(\bar{\mv{p}},\gamma(\bar{\mv{W}},\bar{\mv{V}}))$ to the equations in (\ref{eqn1}) and (\ref{eqn2}).\end{lemma}

\begin{proof}
Note that if the sum-power constraint of all users is considered instead, a similar result to Lemma \ref{lemma3} has been shown in Theorem 1 of \cite{Xu98}. In the following, we extend this result to the case with users' individual power constraints. Suppose that there exist two different solutions to equations (\ref{eqn1}) and (\ref{eqn2}),
denoted by $(\bar{\mv{p}},\gamma(\bar{\mv{W}},\bar{\mv{V}}))$ and
$(\bar{\mv{p}}',\gamma'(\bar{\mv{W}},\bar{\mv{V}}))$, respectively. Define a sequence of
$\theta_k$'s as $\theta_k=\frac{\bar{p}_k'}{\bar{p}_k}$, $\forall
k$. We can without loss of generality re-arrange $\theta_k$'s in a decreasing order
by\begin{align}\label{eqn:sequence}\theta_{t_1}\geq \theta_{t_2} \geq
\cdots \geq \theta_{t_K}.\end{align}Since according to
(\ref{eqn2}) we have $\bar{p}_{k^\ast}=\bar{p}_{k^\ast}'=P_{k^\ast}^{{\rm
max}}$, it follows that $\theta_{k^\ast}=1$ must hold. Hence,
$\theta_{t_1}\geq \theta_{k^\ast}=1$. Moreover, in (\ref{eqn:sequence}), at
least one inequality must hold with a strict inequality sign because otherwise
$\theta_k=1$, $\forall k$, which then implies that only one unique
solution to equations (\ref{eqn1}) and (\ref{eqn2}) exists. Next, we derive the SINR balancing value of $U_{t_1}$ as
follows:\begin{align}\gamma_{t_1}'(\bar{\mv{p}}',\bar{\mv{w}}_{t_1}')&=\frac{\bar{p}_{t_1}'\|\bar{\mv{w}}_{t_1}^H\mv{
h}_{t_1}\|^2}{\bar{\mv{w}}_{t_1}^H\left(\sum\limits_{j\neq
t_1}\bar{p}_j'\mv{h}_{j}\mv{
h}_{j}^H+\sigma^2\mv{I}\right)\bar{\mv{w}}_{t_1}}
\nonumber \\ &=\frac{\bar{p}_{t_1}\|\bar{\mv{w}}_{t_1}^H\mv{
h}_{t_1}\|^2}{\bar{\mv{w}}_{t_1}^H\left(\sum\limits_{j\neq
t_1}\bar{p}_j\mv{h}_{j}\mv{
h}_{j}^H\frac{\theta_j}{\theta_{t_1}}+\sigma^2\mv{I}\frac{1}{\theta_{t_1}}\right)\bar{\mv{w}}_{t_1}}
\nonumber \\ &>\frac{\bar{p}_{t_1}\|\bar{\mv{w}}_{t_1}^H\mv{
h}_{t_1}\|^2}{\bar{\mv{w}}_{t_1}^H\left(\sum\limits_{j\neq
t_1}\bar{p}_j\mv{h}_{j}\mv{
h}_{j}^H+\sigma^2\mv{I}\right)\bar{\mv{w}}_{t_1}}
\nonumber \\ & =\gamma_{t_1}(\bar{\mv{p}},\bar{\mv{w}}_{t_1}).\end{align}Based on
(\ref{equ:equal sinr balancing level}), we have
\begin{align}\label{eqn:greater}\gamma'(\bar{\mv{W}},\bar{\mv{V}})=\gamma'_{t_1}(\bar{\mv{p}}',\bar{\mv{w}}_{t_1}')>\gamma_{t_1}(\bar{\mv{p}},\bar{\mv{w}}_{t_1})=\gamma(\bar{\mv{W}},\bar{\mv{V}}).\end{align}Similarly, we can show that $\gamma'_{t_K}(\bar{\mv{p}}',\bar{\mv{w}}_{t_K}')<\gamma_{t_K}(\bar{\mv{p}},\bar{\mv{w}}_{t_K})$, which
yields\begin{align}\label{eqn:less}\gamma'(\bar{\mv{W}},\bar{\mv{V}})=\gamma'_{t_K}(\bar{\mv{p}}',\bar{\mv{w}}_{t_K}')<\gamma_{t_K}(\bar{\mv{p}},\bar{\mv{w}}_{t_K})=\gamma(\bar{\mv{W}},\bar{\mv{V}}).\end{align}Since
(\ref{eqn:greater}) and (\ref{eqn:less}) contradict to each other,
there must exist one unique solution to equations (\ref{eqn1}) and (\ref{eqn2}). Lemma
\ref{lemma3} is thus proved.
\end{proof}

According to Lemma \ref{lemma3}, there exists a unique solution $(\bar{\mv{p}},\gamma(\bar{\mv{W}},\bar{\mv{V}}))$ to equations (\ref{eqn1}) and (\ref{eqn2}); hence, this solution must be the unique solution that can satisfy (\ref{eqn1}), (\ref{eqn2}), and (\ref{eqn3}) simultaneously, and thus is optimal to problem (\ref{equ:common throughput}). This indicates that given any $\bar{\mv{W}}$ and $\bar{\mv{V}}$, to find the corresponding optimal power and SINR balancing solution to problem (\ref{equ:common throughput}), it is sufficient to study the unique solution to equations (\ref{eqn1}) and (\ref{eqn2}).

Next, we further investigate the properties of equations (\ref{eqn1}) and (\ref{eqn2}). By multiplying both sides of (\ref{eqn1}) by
$\mv{e}_{k^\ast}^T$, we have
\begin{align}\label{equ:power constraint}\frac{\mv{e}_{k^\ast}^T\bar{\mv{p}}}{\gamma(\bar{\mv{W}},\bar{\mv{V}})}=\frac{\bar{P}_{k^\ast}(\bar{\mv{V}},\bar{\tau})}{\gamma(\bar{\mv{W}},\bar{\mv{V}})}=\mv{e}_{k^\ast}^T\mv{D}(\bar{\mv{W}}) \mv \Psi(\bar{\mv{W}})\bar{\mv{p}}+\mv{e}_{k^\ast}^T\mv{D}(\bar{\mv{W}})\mv \sigma(\bar{\mv{W}}).\end{align}
Therefore, by combining (\ref{eqn1}) and (\ref{equ:power
constraint}), it follows that
\begin{align}\label{equ:matrix form2}\frac{1}{\gamma(\bar{\mv{W}},\bar{\mv{V}})}\bar{\mv{p}}_{{\rm ext}}=
\mv{A}_{k^\ast}(\bar{\mv{W}},\bar{\mv{V}})\bar{\mv{p}}_{{\rm ext}},\end{align}where $\bar{\mv{p}}_{{\rm ext}}=\left(\begin{array}{c}\bar{\mv{p}} \\
1\end{array}\right)$ and $\mv{A}_{k^\ast}(\bar{\mv{W}},\bar{\mv{V}})$ is given in (\ref{eqn:matrix Ak}) with $k=k^\ast$.

According to Perron-Frobenius theory \cite{Horn85}, for any
nonnegative matrix, there is at least one positive eigenvalue and
the spectral radius of the matrix is equal to the largest positive
eigenvalue. Furthermore, according to Lemma \ref{lemma3}, there is only one strictly positive eigenvalue to matrix
$\mv{A}_{k^\ast}(\bar{\mv{W}},\bar{\mv{V}})$. Accordingly, it follows from
(\ref{equ:matrix form2}) that given $\bar{\mv{W}}$ and $\bar{\mv{V}}$, the inverse of
the optimal SINR balancing solution $1/\gamma(\bar{\mv{W}},\bar{\mv{V}})$ is the spectral
radius of $\mv{A}_{k^\ast}(\bar{\mv{W}},\bar{\mv{V}})$. In other words, we have
\begin{align}\label{eqn:spectral radius1}
\gamma(\bar{\mv{W}},\bar{\mv{V}})=\frac{1}{\rho(\mv{A}_{k^\ast}(\bar{\mv{W}},\bar{\mv{V}}))}.
\end{align}

Given $\bar{\mv{W}}$ and $\bar{\mv{V}}$, (\ref{eqn:spectral radius1}) relates the optimal SINR balancing solution of problem (\ref{equ:common throughput}) to the spectral radius of the matrix $\mv{A}_{k^\ast}(\bar{\mv{W}},\bar{\mv{V}})$. Finally, we find $k^\ast$ as follows. Note that the optimal power allocation $\bar{\mv{p}}$ and SINR balancing solution $\gamma(\bar{\mv{W}},\bar{\mv{V}})$ to problem (\ref{equ:common throughput}) satisfy (\ref{eqn1}), (\ref{eqn2}), and (\ref{eqn3}). We express the above conditions into $K$ sets of conditions, with the $k$th set of conditions given by
\begin{align}\label{eqn:user k}
\left\{\begin{array}{l}(\ref{eqn1}), \\ \bar{p}_k\leq P_k(\bar{\mv{V}},\bar{\tau}).\end{array}\right.
\end{align}By multiplying both sides of (\ref{eqn1}) by $\mv{e}_k^T$, the power constraint for $U_k$ can be further expressed as
\begin{align}
\frac{\Bar{P}_k(\bar{\mv{V}},\bar{\tau})}{\gamma(\bar{\mv{W}},\bar{\mv{V}})}\geq \frac{\mv{e}_k^T\bar{\mv{p}}}{\gamma(\bar{\mv{W}},\bar{\mv{V}})}=\mv{e}_k^T\mv{D}(\bar{\mv{W}}) \mv \Psi(\bar{\mv{W}})\bar{\mv{p}}+\mv{e}_k^T\mv{D}(\bar{\mv{W}})\mv \sigma(\bar{\mv{W}}).
\end{align}Therefore, (\ref{eqn:user k}) can be equivalently expressed in the matrix form as
\begin{align}\label{eqn:matrix k}
\frac{1}{\gamma(\bar{\mv{W}},\bar{\mv{V}})}\bar{\mv{p}}_{{\rm ext}}\geq
\mv{A}_k(\bar{\mv{W}},\bar{\mv{V}})\bar{\mv{p}}_{{\rm ext}}.
\end{align}Note that (\ref{eqn:matrix k}) holds regardless of $k$.

\begin{lemma}\label{lemma4}{\cite[Theorem 1.6]{Seneta81}}\label{lem:non-negative matrix}
Let $\mv{B}$ be a non-negative irreducible matrix, $\lambda$ a
positive number, and $\mv{x}\geq 0$, $\neq 0$, a vector satisfying
$$\lambda\mv{x}\geq \mv{B}\mv{x}$$ then
$\rho(\mv{B})\leq \lambda$. Moreover, $\rho(\mv{B})=
\lambda$ if and only if $\lambda\mv{x}=\mv{B}\mv{x}$, and in this
case $\mv{x}$ is the dominant eigenvector of $\mv{B}$.
\end{lemma}

According to Lemma \ref{lemma4}, it follows from (\ref{eqn:matrix k}) that
\begin{align}\label{eqn:K inequality}
\frac{1}{\gamma(\bar{\mv{W}},\bar{\mv{V}})}\geq \rho(\mv{A}_k(\bar{\mv{W}},\bar{\mv{V}})), ~~~ \forall k.
\end{align}(\ref{eqn:K inequality}) implies that $\frac{1}{\gamma(\bar{\mv{W}},\bar{\mv{V}})}\geq \max\limits_{1\leq k \leq K}\rho(\mv{A}_k(\bar{\mv{W}},\bar{\mv{V}}))$. According to (\ref{eqn:spectral radius1}), we have
\begin{align}\label{eqn:active power constraint}
k^\ast=\arg \max\limits_{1\leq k \leq K} \rho(\mv{A}_k(\bar{\mv{W}},\bar{\mv{V}})).
\end{align}

By combining (\ref{eqn:spectral radius1}) and (\ref{eqn:active power constraint}), (\ref{eqn:optimal SINR balancing value}) is proved. Moreover, according to (\ref{equ:matrix form2}), if $\bar{\mv{p}}_{{\rm ext}}=\left(\begin{array}{c}\bar{\mv{p}} \\
1\end{array}\right)$ is the dominant eigenvector of $\mv{A}_{k^\ast}(\bar{\mv{W}},\bar{\mv{V}})$, then $\bar{\mv{p}}$ is the optimal power solution to problem (\ref{equ:common throughput}) given $\bar{\mv{W}}$ and $\bar{\mv{V}}$. Theorem \ref{theorem1} is thus proved.

\subsection{Proof of Proposition \ref{proposition1}}\label{appendix2}

Consider the following problem:
\begin{align}\label{eqn:downlink equivalent expression}\mathop{\mathtt{Minimize}}_{\mv{V},\mv{q},\theta} & ~~ \theta \nonumber \\
\mathtt{Subject \ to} & ~~ \mv{A}_k(\bar{\mv{W}},\mv{V})\mv{q}\leq \theta \mv{q}, ~ \forall k, \nonumber \\ & ~~ \mv{q}>\mv{0}, \nonumber \\ & ~~ \sum\limits_{i=1}^l\|\mv{v}_i\|^2\leq P_{{\rm sum}},\end{align}where $\mv{q}=[q_1,\cdots,q_{K+1}]^T$. According to Lemma \ref{lemma4}, the first set of $K$ constraints and $\mv{q}>\mv{0}$ indicate that any feasible solution $(\mv{V},\mv{q},\theta)$ to problem (\ref{eqn:downlink equivalent expression}) satisfies $\theta\geq \rho(\mv{A}_k(\bar{\mv{W}},\mv{V}))$, $\forall k$. In other words, the minimum $\theta$ equals to $\max\limits_{1\leq k \leq K}\rho(\mv{A}_k(\bar{\mv{W}},\mv{V}))$. As a result, problem (\ref{eqn:downlink equivalent expression}) is equivalent to problem (\ref{eqn:downlink}). It can be shown that problem (\ref{eqn:downlink equivalent expression}) can be further expressed in the following form:
\begin{align}\label{eqn:downlink equivalent GP}\mathop{\mathtt{Minimize}}_{\mv{S},\mv{q},\theta} & ~~ \theta \nonumber \\
\mathtt{Subject \ to} & ~~ \sum\limits_{j=1}^K[\mv{X}(\bar{\mv{W}})]_{i,j}\frac{q_j}{q_i \theta}+y_i(\bar{\mv{W}})\frac{q_{K+1}}{q_i \theta} \leq 1, ~ 1\leq i \leq K, \nonumber \\ & ~~ \sum\limits_{j=1}^K[\mv{e}_k^T\mv{X}(\bar{\mv{W}})]_j\frac{q_j}{q_{K+1} \theta}+\mv{e}_k^T\mv{y}(\bar{\mv{W}})\frac{1}{ \theta} \leq \frac{\epsilon \bar{\tau}{\rm Tr}(\mv{G}_k\mv{S})}{1-\bar{\tau}}, ~ 1\leq k \leq K, \nonumber \\ & ~~ {\rm Tr}(\mv{S}) \leq P_{{\rm sum}}, \nonumber \\ & ~~ \mv{S}\succeq \mv{0}. \end{align}

For any scalar $b>0$, let $\tilde{b}=\log b$. Moreover, define $\tilde{\mv{q}}=[\log q_1,\cdots,\log q_{K+1}]^T$, $\forall k$. Then, it can be shown that problem (\ref{eqn:downlink equivalent GP}) is equivalent to problem (\ref{eqn:downlink equivalent convex problem}). Proposition \ref{proposition1} is thus proved.

\subsection{Proof of Theorem \ref{theorem2}}\label{appendix3}
Let $(\tilde{\mv{W}},\tilde{\mv{V}})$ denote the solution obtained by Algorithm \ref{table5}. According to Algorithm \ref{table5}, $(\tilde{\mv{W}},\tilde{\mv{V}})$ satisfies: 1. Given $\mv{V}=\tilde{\mv{V}}$, $\tilde{\mv{W}}$ is the optimal solution to problem (\ref{eqn:uplink}); and 2. Given $\mv{W}=\tilde{\mv{W}}$, $\tilde{\mv{V}}$ is the optimal solution to problem (\ref{eqn:downlink}). Furthermore, define $\tilde{k}^\ast=\arg \max\limits_{1\leq k \leq K} \rho(\mv{A}_k(\tilde{\mv{W}},\tilde{\mv{V}}))$, and $\tilde{\mv{p}}_{{\rm ext}}=\left(\begin{array}{c}\tilde{\mv{p}} \\ 1 \end{array} \right)$ as the dominant eigenvector of the matrix $\mv{A}_{\tilde{k}^\ast}(\tilde{\mv{W}},\tilde{\mv{V}})$; then $\tilde{\mv{p}}=(\tilde{p}_1,\cdots,\tilde{p}_K)$ is the optimal power solution to problem (\ref{equ:common throughput}) given $\mv{W}=\tilde{\mv{W}}$ and $\mv{V}=\tilde{\mv{V}}$ according to Theorem \ref{theorem1}.

\begin{lemma}\label{lemma8}{\cite[Corollary 8.3.3]{Horn85}} For any non-negative irreducible $K$-dimension matrix $\mv{B}$, its spectral radius can be expressed as
\begin{align}
\rho(\mv{B})=\max\limits_{\mv{y}\geq \mv{0}, \mv{y} \neq \mv{0}} \min\limits_{1\leq j \leq K} \frac{\mv{e}_j^T\mv{B}\mv{y}}{\mv{e}_j^T\mv{y}}.
\end{align}
\end{lemma}

Let $\rho^\ast$ denote the optimal value of problem (\ref{eqn:equivalent problem}). According to Lemma \ref{lemma8}, it follows that
\begin{align}\label{eqn:a1}
\rho^\ast=\min\limits_{\mv{W}}\min\limits_{\mv{V}\in \mathcal{V}}\max\limits_{1\leq k \leq K}\max\limits_{\mv{y}_k\geq \mv{0}, \mv{y}_k \neq \mv{0}} \min\limits_{1\leq j_k \leq K+1} \frac{\mv{e}_{j_k}^T\mv{A}_k(\mv{W},\mv{V})\mv{y}_k}{\mv{e}_{j_k}^T\mv{y}_k},
\end{align}where $\mathcal{V}=\{\mv{V}|\sum\limits_{i=1}^l\|\mv{v}_i\|^2\leq P_{{\rm sum}}\}$.

First, we assume that $\mv{y}_k=\tilde{\mv{p}}_{{\rm ext}}$, $\forall k$. Then define $\bar{\rho}^\ast$ as
\begin{align}\label{eqn:a2}
\bar{\rho}^\ast=\min\limits_{\mv{W}}\min\limits_{\mv{V}\in \mathcal{V}}\max\limits_{1\leq k \leq K}\min\limits_{1\leq j_k \leq K+1} \frac{\mv{e}_{j_k}^T\mv{A}_k(\mv{W},\mv{V})\tilde{\mv{p}}_{{\rm ext}}}{\mv{e}_{j_k}^T\tilde{\mv{p}}_{{\rm ext}}}.
\end{align}It can be observed that $\bar{\rho}^\ast$ is a lower bound of $\rho^\ast$, i.e., $\bar{\rho}^\ast\leq \rho^\ast$.

According to the definition of $\mv{A}_k(\mv{W},\mv{V})$'s given in (\ref{eqn:matrix Ak}), we have
\begin{align}\label{eqn:a3}
\frac{\mv{e}_{j_k}^T\mv{A}_k(\mv{W},\mv{V})\tilde{\mv{p}}_{{\rm ext}}}{\mv{e}_{j_k}^T\tilde{\mv{p}}_{{\rm ext}}}=\left\{\begin{array}{ll}\frac{1}{\gamma_{j_k}(\tilde{\mv{p}},\mv{w}_{j_k})}, & {\rm if} ~ 1\leq j_k \leq K, \\ \frac{\tilde{p}_k}{\bar{P}_k(\mv{V},\bar{\tau})}\times \frac{1}{\gamma_k(\tilde{\mv{p}},\mv{w}_k)}, & {\rm if} ~ j_k=K+1, \end{array} \right.
\end{align}where $\gamma_k(\tilde{\mv{p}},\mv{w}_k)$ and $\bar{P}_k(\mv{V},\bar{\tau})$ are given in (\ref{equ:SINR in
SIMO-IC}) and (\ref{eq:available power}), respectively, $\forall k$. It is worth noting that $\tilde{\mv{W}}$ is the optimal MMSE receiver corresponding to the power allocation $\tilde{\mv{p}}$, as shown in \cite{Schubert04}, \cite{Zhanglan08}, which maximizes $\gamma_k(\tilde{\mv{p}},\mv{w}_k)$, $\forall k$. As a result, according to (\ref{eqn:a3}), given any $\mv{V}$ we have
\begin{align}\label{eqn:a4}
\frac{\mv{e}_{j_k}^T\mv{A}_k(\tilde{\mv{W}},\mv{V})\tilde{\mv{p}}_{{\rm ext}}}{\mv{e}_{j_k}^T\tilde{\mv{p}}_{{\rm ext}}}\leq \frac{\mv{e}_{j_k}^T\mv{A}_k(\mv{W},\mv{V})\tilde{\mv{p}}_{{\rm ext}}}{\mv{e}_{j_k}^T\tilde{\mv{p}}_{{\rm ext}}} ~~ {\rm if} ~ \mv{W}\neq \tilde{\mv{W}}, ~~ \forall k, ~~ \forall j_k.
\end{align}It then follows
\begin{align}\label{eqn:a5}
\min\limits_{\mv{W}}\max\limits_{1\leq k \leq K}\min\limits_{1\leq j_k \leq K+1}\frac{\mv{e}_{j_k}^T\mv{A}_k(\mv{W},\mv{V})\tilde{\mv{p}}_{{\rm ext}}}{\mv{e}_{j_k}^T\tilde{\mv{p}}_{{\rm ext}}}=\max\limits_{1\leq k \leq K}\min\limits_{1\leq j_k \leq K+1}\frac{\mv{e}_{j_k}^T\mv{A}_k(\tilde{\mv{W}},\mv{V})\tilde{\mv{p}}_{{\rm ext}}}{\mv{e}_{j_k}^T\tilde{\mv{p}}_{{\rm ext}}}, ~~ \forall \mv{V}.
\end{align}Since (\ref{eqn:a5}) holds for all $\mv{V}$, it follows that
\begin{align}\label{eqn:a6}
\bar{\rho}^\ast=\min\limits_{\mv{W}}\min\limits_{\mv{V}\in \mathcal{V}}\max\limits_{1\leq k \leq K}\min\limits_{1\leq j_k \leq K+1} \frac{\mv{e}_{j_k}^T\mv{A}_k(\mv{W},\mv{V})\tilde{\mv{p}}_{{\rm ext}}}{\mv{e}_{j_k}^T\tilde{\mv{p}}_{{\rm ext}}}=\min\limits_{\mv{V}\in \mathcal{V}}\max\limits_{1\leq k \leq K}\min\limits_{1\leq j_k \leq K+1}\frac{\mv{e}_{j_k}^T\mv{A}_k(\tilde{\mv{W}},\mv{V})\tilde{\mv{p}}_{{\rm ext}}}{\mv{e}_{j_k}^T\tilde{\mv{p}}_{{\rm ext}}}.
\end{align}Note that given $\mv{W}=\tilde{\mv{W}}$ and $\mv{V}=\tilde{\mv{V}}$, $\tilde{\mv{p}}_{{\rm ext}}$ is the dominant eigenvector of $\mv{A}_{\tilde{k}^\ast}(\tilde{\mv{W}},\tilde{\mv{V}})$, i.e.,
\begin{align}\label{eqn:a7}
\mv{A}_{\tilde{k}^\ast}(\tilde{\mv{W}},\tilde{\mv{V}})\tilde{\mv{p}}_{{\rm ext}}=\rho(\mv{A}_{\tilde{k}^\ast}(\tilde{\mv{W}},\tilde{\mv{V}}))\tilde{\mv{p}}_{{\rm ext}}.
\end{align}We thus have $\gamma_k(\tilde{\mv{p}},\tilde{\mv{w}}_k)=\frac{1}{\rho(\mv{A}_{\tilde{k}^\ast}(\tilde{\mv{W}},\tilde{\mv{V}}))}$, $\forall k$. As a result, with $\mv{W}=\tilde{\mv{W}}$, (\ref{eqn:a3}) can be further simplified as
\begin{align}\label{eqn:a8}
\frac{\mv{e}_{j_k}^T\mv{A}_k(\tilde{\mv{W}},\mv{V})\tilde{\mv{p}}_{{\rm ext}}}{\mv{e}_{j_k}^T\tilde{\mv{p}}_{{\rm ext}}}=\left\{\begin{array}{ll}\rho(\mv{A}_{k^\ast}(\tilde{\mv{W}},\tilde{\mv{V}})), & {\rm if} ~ 1\leq j_k \leq K, \\ \frac{\tilde{p}_k}{\bar{P}_k(\mv{V},\bar{\tau})}\times \rho(\mv{A}_{k^\ast}(\tilde{\mv{W}},\tilde{\mv{V}})), & {\rm if} ~ j_k=K+1. \end{array} \right.
\end{align}

Next, consider the special case of $\mv{V}=\tilde{\mv{V}}$. Since given $\tilde{\mv{W}}$ and $\tilde{\mv{V}}$, $\tilde{\mv{p}}$ is the optimal power solution to problem (\ref{equ:common throughput}), (\ref{eqn2}) and (\ref{eqn3}) must hold, i.e., $\tilde{p}_k= \bar{P}_k(\tilde{\mv{V}},\bar{\tau})$ if $k=\tilde{k}^\ast$, and $\tilde{p}_k\leq \bar{P}_k(\tilde{\mv{V}},\bar{\tau})$ otherwise. As a result, it follows that
\begin{align}\label{eqn:a9}
\min\limits_{1\leq j_k \leq K+1} \frac{\mv{e}_{j_k}^T\mv{A}_k(\tilde{\mv{W}},\tilde{\mv{V}})\tilde{\mv{p}}_{{\rm ext}}}{\mv{e}_{j_k}^T\tilde{\mv{p}}_{{\rm ext}}}=\left\{\begin{array}{ll}\rho(\mv{A}_{\tilde{k}^\ast}(\tilde{\mv{W}},\tilde{\mv{V}})), & {\rm if} ~ k=\tilde{k}^\ast, \\ \frac{\tilde{p}_k}{\bar{P}_k(\tilde{\mv{V}},\bar{\tau})}\times \rho(\mv{A}_{\tilde{k}^\ast}(\tilde{\mv{W}},\tilde{\mv{V}})), & {\rm if} ~ k\neq \tilde{k}^\ast. \end{array} \right.
\end{align}Thus, we have $\max\limits_{1\leq k \leq K}\min\limits_{1\leq j_k \leq K+1} \frac{\mv{e}_{j_k}^T\mv{A}_k(\tilde{\mv{W}},\tilde{\mv{V}})\tilde{\mv{p}}_{{\rm ext}}}{\mv{e}_{j_k}^T\tilde{\mv{p}}_{{\rm ext}}}=\rho(\mv{A}_{\tilde{k}^\ast}(\tilde{\mv{W}},\tilde{\mv{V}}))$ because $\frac{\tilde{p}_k}{\bar{P}_k(\tilde{\mv{V}},\bar{\tau})}\leq 1$ if $k\neq \tilde{k}^\ast$. According to (\ref{eqn:a6}), it thus follows that
\begin{align}\label{eqn:a10}
\bar{\rho}^\ast=\min\limits_{\mv{V}\in \mathcal{V}}\max\limits_{1\leq k \leq K}\min\limits_{1\leq j_k \leq K+1}\frac{\mv{e}_{j_k}^T\mv{A}_k(\tilde{\mv{W}},\mv{V})\tilde{\mv{p}}_{{\rm ext}}}{\mv{e}_{j_k}^T\tilde{\mv{p}}_{{\rm ext}}}\leq \max\limits_{1\leq k \leq K}\min\limits_{1\leq j_k \leq K+1}\frac{\mv{e}_{j_k}^T\mv{A}_k(\tilde{\mv{W}},\tilde{\mv{V}})\tilde{\mv{p}}_{{\rm ext}}}{\mv{e}_{j_k}^T\tilde{\mv{p}}_{{\rm ext}}}=\rho(\mv{A}_{\tilde{k}^\ast}(\tilde{\mv{W}},\tilde{\mv{V}})).
\end{align}

Next, we show $\bar{\rho}^\ast=\rho(\mv{A}_{k^\ast}(\tilde{\mv{W}},\tilde{\mv{V}}))$ by contradiction. Assume that $\bar{\rho}^\ast<\rho(\mv{A}_{\tilde{k}^\ast}(\tilde{\mv{W}},\tilde{\mv{V}}))$. In this case, there exists at least a $\mv{V}=\mv{V}'$ such that $\max\limits_{1\leq k \leq K}\min\limits_{1\leq j_k \leq K+1}\frac{\mv{e}_{j_k}^T\mv{A}_k(\tilde{\mv{W}},\mv{V}')\tilde{\mv{p}}_{{\rm ext}}}{\mv{e}_{j_k}^T\tilde{\mv{p}}_{{\rm ext}}}<\rho(\mv{A}_{\tilde{k}^\ast}(\tilde{\mv{W}},\tilde{\mv{V}}))$. According to (\ref{eqn:a8}), it follows that $\tilde{p}_k< \bar{P}_k(\mv{V}',\bar{\tau})$, $\forall k$. This indicates that
\begin{align}\label{eqn:a11}
\mv{A}_k(\tilde{\mv{W}},\mv{V}') \tilde{\mv{p}}_{{\rm ext}} \leq \rho(\mv{A}_{\tilde{k}^\ast}(\tilde{\mv{W}},\tilde{\mv{V}})) \tilde{\mv{p}}_{{\rm ext}} \ {\rm but} \ \neq \rho(\mv{A}_{\tilde{k}^\ast}(\tilde{\mv{W}},\tilde{\mv{V}})) \tilde{\mv{p}}_{{\rm ext}}, \ \forall k.
\end{align}According to Lemma \ref{lemma4}, it follows from (\ref{eqn:a11}) that $\rho(\mv{A}_k(\tilde{\mv{W}},\mv{V}'))<\rho(\mv{A}_{\tilde{k}^\ast}(\tilde{\mv{W}},\tilde{\mv{V}}))$, $\forall k$. In other words, we have  $\max\limits_{1\leq k \leq K} \rho(\mv{A}_k(\tilde{\mv{W}},\mv{V}')) <\rho(\mv{A}_{\tilde{k}^\ast}(\tilde{\mv{W}},\tilde{\mv{V}}))$, which contradicts to the fact that given $\tilde{\mv{W}}$, $\tilde{\mv{V}}$ is the optimal solution to problem (\ref{eqn:downlink}). Therefore, we have $\bar{\rho}^\ast=\rho(\mv{A}_{k^\ast}(\tilde{\mv{W}},\tilde{\mv{V}}))$.

Last, by combining (\ref{eqn:a1}) and (\ref{eqn:a2}), we have $\rho^\ast\geq \bar{\rho}^\ast=\rho(\mv{A}_{\tilde{k}^\ast}(\tilde{\mv{W}},\tilde{\mv{V}}))=\max\limits_{1\leq k \leq K} \rho(\mv{A}_k(\tilde{\mv{W}},\tilde{\mv{V}}))$. Moreover, $\rho^\ast=\min\limits_{\mv{W}}\min\limits_{\mv{V}\in \mathcal{V}}\max\limits_{1\leq k \leq K}\rho(\mv{A}_k(\mv{W},\mv{V}))\leq \max\limits_{1\leq k \leq K} \rho(\mv{A}_k(\tilde{\mv{W}},\tilde{\mv{V}}))$ also holds. To summarize, we have $\rho^\ast=\rho(\mv{A}_{\tilde{k}^\ast}(\tilde{\mv{W}},\tilde{\mv{V}}))$, i.e., the solution $(\tilde{\mv{W}},\tilde{\mv{V}})$ achieves the optimal value of problem (\ref{eqn:equivalent problem}). Theorem \ref{theorem2} is thus proved.

\end{appendix}

\end{document}